\documentclass[conference,compsoc]{IEEEtran}

\usepackage{etoolbox}
\AtBeginEnvironment{quote}{\singlespacing\small}
\usepackage{enumerate}
\usepackage[usenames,dvipsnames]{xcolor}
\usepackage[pass]{geometry}
\usepackage{amsfonts}
\usepackage[usenames,dvipsnames]{xcolor}
\usepackage{url}
\usepackage{graphicx}
\usepackage{balance}
\usepackage{comment}
\usepackage[tight,footnotesize]{subfigure}

\usepackage{multirow}
\usepackage{array}
\usepackage[justification=centering]{caption}
\usepackage{algorithm}
\usepackage{algpseudocode}
\usepackage{graphics}
\usepackage{epsfig}
\usepackage{amsmath}
\usepackage{amsthm}
\usepackage{float}
\usepackage{setspace}
\usepackage{color}
\usepackage{cite}

\renewcommand\thesection{\arabic{section}}
\renewcommand\thesubsection{\thesection.\arabic{subsection}}
\renewcommand\thesubsubsection{\thesubsection.\arabic{subsubsection}}

\renewcommand\thesectiondis{\arabic{section}}
\renewcommand\thesubsectiondis{\thesectiondis.\arabic{subsection}}
\renewcommand\thesubsubsectiondis{\thesubsectiondis.\arabic{subsubsection}}

\newtheorem{theorem}{Theorem}

\newtheorem{lemma}{Lemma}

\newtheorem{example}{Example}

\newtheorem{definition}{Definition}

\newtheorem*{test*}{Schedulability Test in~\cite{li2013outstanding}}
\newtheorem*{test1*}{Schedulability Test in~\cite{chen2014capacity}}

\IEEEoverridecommandlockouts

\newcommand{\cp}{Lazy-Cpath }

\begin{document}


\title{New Analysis Techniques for Supporting Hard Real-Time Sporadic DAG Task Systems on Multiprocessors}

\author{Zheng Dong and Cong Liu\\
Department of Computer Science, University of Texas at Dallas}
\maketitle

\thispagestyle{plain}
\pagestyle{plain}

\noindent\textbf{Abstract.} The scheduling and schedulability analysis of real-time directed acyclic graph (DAG) task systems have received much recent attention. The DAG model can accurately represent intra-task parallelism and precedence constraints existing in many application domains. Existing techniques show that analyzing the DAG model is fundamentally more challenging compared to the ordinary sporadic task model, due to the complex intra-DAG precedence constraints which may cause rather pessimistic schedulability loss. However, such increased loss is counter-intuitive because the DAG structure shall better exploit the parallelism provided by the multiprocessor platform. Our observation is that the intra-DAG precedence constraints, if not carefully considered by the scheduling algorithm, may cause very unpredictable execution behaviors of subtasks in a DAG and further cause pessimistic analysis.  In this paper, we present a set of novel scheduling and analysis techniques for better supporting hard real-time sporadic DAG tasks on multiprocessors, through smartly defining and analyzing the execution order of subtasks in each DAG.
Evaluation demonstrates that our developed utilization-based schedulability test is highly efficient, which dramatically improves schedulability of existing utilization-based tests by over 60\% on average. Interestingly, when each DAG in the system is an ordinary sporadic task, our test becomes identical to the classical density test designed for the sporadic task model.




\section{Introduction}
In many real-time and embedded systems, applications are defined using processing graphs~\cite{kumar1994introduction} to better exploit the parallel computing capability provided by the multicore hardware. The scheduling and schedulability analysis of the real-time DAG (directed acyclic graph) model, which is defined to accurately capture such graph structure and intra-graph precedence constraints, have received much recent attention~\cite{saifullah2014parallel, ferry2013real, fonseca2017improved, liu2010supporting, jiang2017semi, saifullah2013multi, jiang2016decomposition, qamhieh2013global, nelissen2012techniques, kim2013parallel, andersson2012analyzing, bonifaci2013feasibility, li2013outstanding, baruah2014improved, chen2014capacity, li2014analysis, baruah2015federated, baruah2015federated2, baruah2015federated3, yang2016reducing, li2017mixed, guo2017energy, li2014federated}.  Although the development of such works represents a major and promising step towards better supporting real-time DAG tasks on multiprocessors, they fundamentally suffer from intra-DAG precedence constraints and  may exhibit pessimistic schedulability loss. By comparing such DAG-based tests with the tests developed for ordinary sporadic tasks, it is not hard to observe that the schedulability loss becomes much more significant (also directly noted in~\cite{li2013outstanding, chen2014capacity}).

However, such increased schedulability loss is counter-intuitive. If a DAG task system contains the same amount of workload as an ordinary sporadic task system, then the DAG task system shall be easier to be schedulable because a DAG task may better exploit the parallelism provided by multiprocessors. Fig.~\ref{fig:introexample} illustrates a simple example illustrating this intuition. As seen in Fig.~\ref{fig:introexample}(a), $\tau_{3,1}$ misses its deadline if $\tau_3$ is a sporadic task with an execution cost of three time units. Interestingly, if $\tau_3$ is a DAG task (Fig.~\ref{fig:introexample}(c)) with the same total execution cost, then it becomes schedulable as seen in Fig.~\ref{fig:introexample}(b). This is because as a DAG task, $\tau_3$ can benefit from the parallelism provided by the  two-processor system.

This intuition motivated us to understand why the existing DAG-based schedulability analysis techniques~\cite{li2013outstanding, chen2014capacity} indicate the opposite conclusion where DAG task systems are harder to be schedulable and encounter more schedulability loss compared to ordinary sporadic tasks. We observe that a common consensus is that for ordinary sporadic tasks, a released job can start immediately whenever a processor becomes idle; yet for DAG tasks, a released jobs from a subtask in a DAG can start only if all its predecessors have completed, besides the premise of processor availability. This extra condition causes over pessimism in the analysis as the execution behaviors of subtasks in a DAG is unpredictable and complex due to the precedence constraints among subtasks defined by the DAG structure. Existing scheduling algorithms are not smart enough to handle such complex DAG structures (also noted in \cite{li2013outstanding}).

\begin{figure}[t]
{\includegraphics[width=0.48\textwidth]{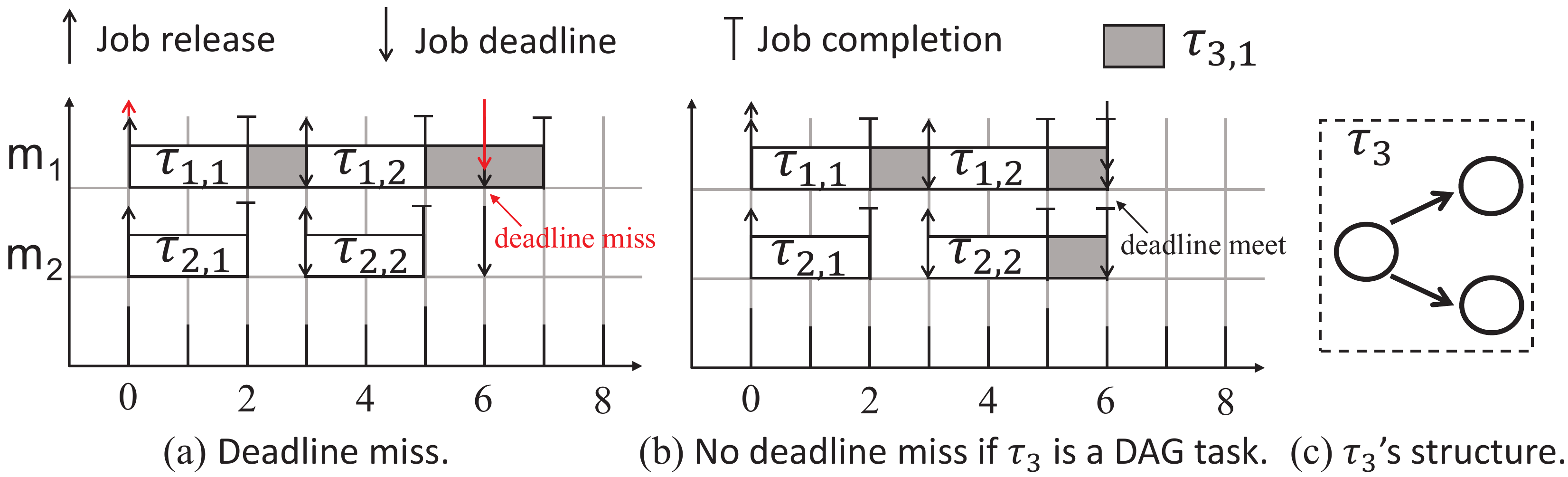}
}\caption{\footnotesize{Example illustrating that the DAG model may better exploit multiprocessor parallelism  compared to the sporadic task model. $\tau_1$ and $\tau_2$ are  sporadic tasks each with an execution cost of two time units.}
}\label{fig:introexample}
\end{figure}

In this paper, we propose a set of novel scheduling methods and analysis techniques that enable the system to take the intuitive advantage of the parallelism benefits for executing DAG tasks on multiprocessors. There are two major components in our proposed approach: a \cp policy and a new executing/non-executing interval-based analysis technique. The intuitive idea behind the \cp policy is to smartly define the execution order of ready jobs of subtasks belonging to the same DAG. By applying the \cp policy, intuitively, the resulting GEDF schedule would allow us to analyze the DAG tasks while mostly ignoring the complex intra-DAG precedence constraints. Combining the \cp policy with a novel executing/non-executing interval-based analysis technique, our developed analysis can more precisely characterize the workload distribution in the analysis window of interest, which yields rather efficient schedulability test.

\noindent\textbf{Overview of related work.} Scheduling and schedulability analysis of real-time DAG task systems have received much recent attention~\cite{jiang2017semi, saifullah2013multi, jiang2016decomposition, qamhieh2013global, nelissen2012techniques, kim2013parallel, andersson2012analyzing, bonifaci2013feasibility, li2013outstanding, baruah2014improved, chen2014capacity, li2014analysis, baruah2015federated, baruah2015federated2, baruah2015federated3}. Capacity augmentation bounds have been derived under  decomposition-based scheduling~\cite{qamhieh2013global, nelissen2012techniques, kim2013parallel, saifullah2013multi, jiang2016decomposition} where DAGs are transformed into sequential subtasks and scheduled under traditional schedulers, non-decomposition-based scheduling~\cite{andersson2012analyzing, bonifaci2013feasibility, li2013outstanding, chen2014capacity, baruah2014improved} where subtasks in DAGs are scheduled according to their corresponding precedence constraints, and federated scheduling~\cite{li2014analysis, baruah2015federated, baruah2015federated2, baruah2015federated3} where high-utilization DAGs are exclusively executed on dedicated processors and low-utilization DAGs are treated as sequential sporadic tasks and share the remaining processors. In \cite{li2013outstanding, chen2014capacity}, two utilization-based schedulability tests have also been proposed under GEDF scheduling. As directly noted in \cite{li2013outstanding}, these schedulability tests may be pessimistic due to the sub-optimality of directly applying GEDF to schedule the DAG task system.

\vspace{1mm}
\noindent\textbf{Our contributions.} In this paper, we develop a utilization-based schedulability test for hard real-time (HRT)  sporadic DAG task systems scheduled on a multiprocessor with $M$ identical cores under GEDF. Specifically, we show that any HRT DAG task system $\tau$ is schedulable under GEDF with the \cp policy, if for every task $\tau_k \in \tau$,
\begin{equation}\label{eq:finaltest2}
\sum_{i=1}^n\eta_i \leq M - (M - 1)\times\sigma_k,
\end{equation}
\noindent holds, where $\eta_i$ is defined in Eq.~\ref{eq:eta} and $\sigma_k$ denotes the total utilization of subtasks on the critical path (defined in Def.~\ref{definition:cputilization}) of DAG $\tau_k$. This schedulability test can be view as the DAG version of the schedulability test derived for the ordinary sporadic task model given in~\cite{baker2003},  through replacing the constraint of $\sigma_k$ in Eq.~\ref{eq:finaltest2} by $\tau_k$'s utilization.
To derive this test, we invent a novel \cp policy that advises GEDF to execute the ready job of any subtask on the corresponding DAG's critical path as the last job among all ready jobs at any time instant. By further inventing and applying a new executing/non-executing interval-based anylsis technique to the resulting GEDF schedule with the \cp policy enabled, we are able to more precisely characterize the workload distribution in an analysis window.

Evaluation demonstrates that the above schedulability test is highly efficient, which dramatically improves schedulability compared to existing utilization-based tests by over 60\% on average. Interestingly, this test can be generalized to the classical density test\footnote{Under the density test, an implicit-deadline sporadic task system is schedulable under GEDF if $U_{sum} \leq M - (M-1) \cdot u_{max}$ holds, where $u_{max}$ denotes the maximum task utilization in the system.}~\cite{goossens2003priority} designed for ordinary sporadic task systems. When each DAG task only contains a single subtask, the schedulability condition Eq.~\ref{eq:finaltest2} becomes identical to the density test. 


\section{System Model}
\label{sec:model}

We consider the problem of scheduling a set $\tau = \{\tau_1, \dots, \tau_n\}$ of $n$ independent sporadic DAG tasks on $M$ identical processors. Each task $\tau_i$ is specified by a tuple $(G_i,d_i, p_i)$, where $G_i = (V_i,E_i)$ is a DAG whose structure is defined by a set of subtasks $V_i$ and a set of edges $E_i$ connecting these subtasks, and $d_i$ and $p_i$ denote the relative deadline and period of DAG $\tau_i$, respectively. Let $\tau_i^k$ denote the $k^{th}$ \emph{subtask} of $\tau_i$ and $\tau_{i,j}^k$ denote the $j^{th}$ job of $\tau_i^k$. A directed edge from subtask $\tau_i^k$ to $\tau_i^h$ implies that $\tau_i^k$ is a predecessor subtask of $\tau_i^h$ (also $\tau_i^h$ is a successor subtask of $\tau_i^k$). Thus, for any $j$, job $\tau_{i,j}^k$ is a predecessor job of $\tau_{i,j}^h$.  Any job of any subtask is \emph{ready} to start executing if all of its predecessor jobs have completed.
Each sporadic DAG task $\tau_i$ generates an infinite sequence of \textit{dag-jobs}, with arrival times of successive jobs separated by at least $p_i$ time units. Clearly, each dag-job released by $\tau_i$ is composed by $|V_i|$ jobs belonging to the $|V_i|$ subtasks. Thus, the $j^{th}$ jobs of  all subtasks in $\tau_i$ share a same release time denoted by $r_{i,j}$ and a same absolute deadline denoted by $d_{i,j} = r_{i,j} + d_i$.
 Each subtask has a worst case execution time $e_i^k$. Our DAG model allows multiple source subtasks (i.e., those with no predecessor subtask) and sink subtask (i.e., with no successor subtask) are allowed in our general DAG model.
In this paper, we focus our attention on implicit-deadline DAG task systems, where $d_i = p_i$ holds for all DAG tasks.

\begin{definition}\label{def:critical}
\textbf{(Critical Path of a DAG)} A chain in DAG $\tau_i$ is a sequence of subtasks $\tau_i^{x_1}, \tau_i^{x_2}, \dots, \tau_i^{x_z} \in V_i$ such that $(\tau_i^{x_u}, \tau_i^{x_{u+1}})$ is an edge in $G_i$, $1 \leq u < z-1$. The length of this chain is defined to be the sum of the execution times of all the subtasks on the chain: $\sum_{u=1}^z e_i^{x_u}$. We denote by $\mathbf{Cpath}_i = \{\tau_i^{x_1}, \tau_i^{x_2}, \dots, \tau_i^{x_k}\}$ the longest chain, named tje \textit{critical path} of DAG $\tau_i$. Let $\mathbf{len}(\mathbf{Cpath}_i)$ denote the length of $\mathbf{Cpath}_i$. Each DAG has only one critica path. Note that there may exist multiple paths with the same maximum length among all paths in a DAG task. In such cases, we arbitrarily select one from them as the only critical path of this DAG.~\footnote{Note that $\mathbf{Cpath}_i$ can be found in time linear in the number of vertices and the number of edges in $\tau_i$, by first obtaining a topological order of the vertices of the graph and then running a straightforward dynamic program~\cite{bonifaci2013feasibility}.}
\end{definition}

\begin{definition}\label{definition:cputilization}
\textbf{(cp-utilization)} For each DAG $\tau_i$, we define $\sigma_i = \mathbf{len}(\mathbf{Cpath}_i)/p_i$ as the cp-utilization of $\tau_i$, which represents the total utilization of subtasks on the critical path of $\tau_i$.
\end{definition}

We now define several additional terminologies for the DAG task model:
\begin{itemize}
\item We denote by $C_i = \sum_{v=1}^{|V_i|}e_i^v$ the total execution time of each DAG  $\tau_i$.
\item For each DAG $\tau_i$ we define a utilization $u_i = C_i/p_i$. The total utilization of the DAG task system is $U_{sum} = \sum_{i=1}^n u_i$.
\item For a sporadic DAG task system $\tau$, we define its maximum cp-utilization $\sigma_{max}(\tau)$ to be the largest cp-utilization of any task in $\tau$: $\sigma_{max}(\tau) = \max_{\tau_i\in\tau}\sigma_i$.
\end{itemize}

We require $\mathbf{len}(\mathbf{Cpath}_i)\leq p_i$, and $U_{sum} \leq M$; otherwise, deadlines will be missed. Successive dag-jobs released by the same DAG are required to execute in sequence. We study the DAG task system under preemptive GEDF scheduling: dag-jobs with earlier deadlines are assigned higher priorities. Note that all jobs releaded by the subtasks in a DAG inherit the same priority assigned to the corresponding dag-job. We assume that ties are broken by task ID (lower IDs are favored). Throughout the paper, we assume that time is integral. Thus, a job that executes at time instant $t$ executes during the entire time interval $[t, t + 1)$.


\begin{figure}[H]
\centerline
{\includegraphics[width=0.48\textwidth]{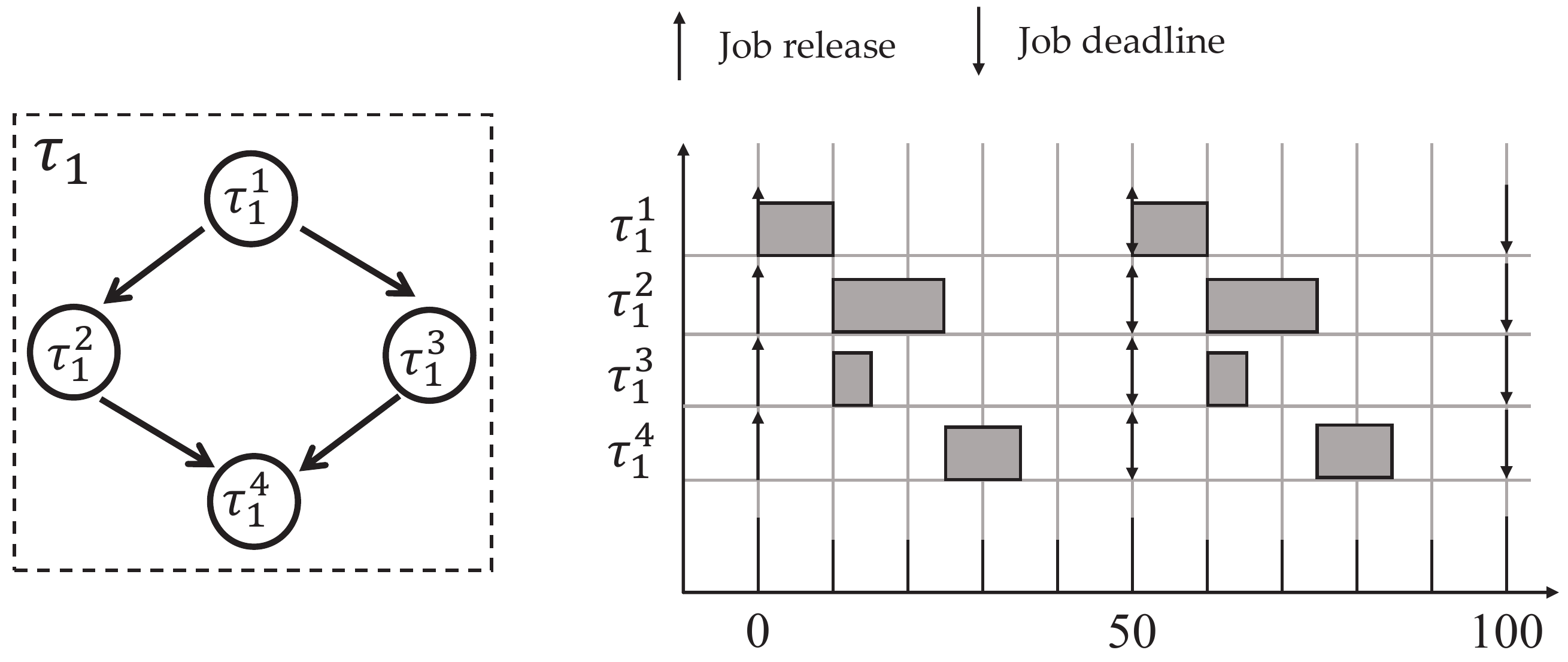}
}\caption{\footnotesize{Example DAG task.}}\label{fig:example}
\end{figure}

\begin{example}
Figure~\ref{fig:example} shows an example of scheduling a DAG task $\tau_1$ with a deadline of 50 time units on a two-processor system consisting of four subtasks, $\tau_1^1$, $\tau_1^2$, $\tau_1^3$ and $\tau_1^4$, with a worst-case execution time of $10$, $15$, $5$ and $10$ time units, respectively. For this DAG, its critical path is $\tau_1^1$, $\tau_1^2$, $\tau_1^4$, its cp-length is $\mathbf{len}(\mathbf{Cpath}_1)= 35$, and its cp-utilization is $\sigma_1 = 0.7$.
\end{example}



\section{Subtask Ordering: The Lazy-Cpath Policy}\label{sec:policy}

Due to intra-DAG precedence constraints, analyzing the schedulability of DAG task systems on a multiprocessor could be more challenging compared to the case of ordinary sporadic tasks. Precedence constraints among subtasks in a DAG make it hard to precisely analyze the subtasks' execution behavior and bound the interference workload upon certain analyzed jobs, as shown in several recent works~\cite{li2013outstanding, chen2014capacity, baruah2014improved, bonifaci2013feasibility}. For instance, a recent work~\cite{li2013outstanding} has demonstrated this challenge by assuming worst-case interference scenarios due to precedence constraints and developing the following utilization-based schedulability test for a sporadic DAG task system. The following quotes can be found in the abstract of~\cite{li2013outstanding}:

\begin{quote}
``\emph{For the proposed capacity augmentation bound of $4 - \frac{2}{M}$ for implicit deadline tasks under GEDF, we prove that if a task set has a total utilization of at most $\frac{M}{4-\frac{2}{M}}$ and each task's critical path length is no more than $\frac{1}{4-\frac{2}{M}}$ of its deadline, it can be scheduled on a machine with $M$ processors under GEDF. For the standard resource augmentation bound of $2 - \frac{1}{M}$ for arbitrary deadline tasks under GEDF, we prove that if an ideal optimal scheduler can schedule a task set on $M$ unit-speed processors, then GEDF can schedule the same task set on $M$ processors of speed $2-\frac{1}{M}$. However, this bound does not lead to a schedulabilty test since the ideal optimal scheduler is only hypothetical and is not known.}''
\end{quote}

When tasks are scheduled on a multiprocessor platform ($M\geq 2$), the above schedulability test~\cite{li2013outstanding} essentially requires that the total system utilization is no greater than $\frac{M}{3}$ and each task's critical path length is no more than $\frac{1}{3}$ of its deadline. Such constraints are rather pessimistic.

As quoted above, designing a good scheduler is critical to develop a corresponding efficient schedulability test. Our goal is thus to design a smarter runtime scheduler that efficiently supports DAG task scheduling on multiprocessors. In the rest of this section, we propose a Lazy-Cpath policy to precisely define the execution order of subtasks in a DAG. By combing the \cp policy with GEDF, the resulting scheduler yields several important properties about the execution behavior of the DAG tasks, which enable us to ultimately develop significantly improved schedulability tests.

\subsection{The Lazy-Cpath Policy}

We now formally define the \cp Policy and derive its  beneficial properties that enable us to develop an efficient schedulability test.



\noindent\textbf{The \cp policy:} Under the \cp policy, at any time instance, ready jobs of subtasks on a DAG's critical path has the lowest priority among all ready jobs of subtasks belonging to the same DAG. Note that for any DAG, at most one subtask on its critical path could have a ready job, because subtasks on any DAG's critical path form a single chain of subtasks whose released jobs must be executed sequentially. Thus, under the \cp policy, GEDF works as follows: at each time instant, the scheduler first tries to schedule as many ready jobs with the earliest deadlines; if the number of such ready jobs is larger than $M$, then GEDF first schedules ready jobs of subtasks that are not on the critical path of the corresponding DAG. Let CP-GEDF denote the scheduling strategy after applying the \cp policy to GEDF.

According to the above definition, the \cp policy does not define specific execution order among jobs of subtasks that are not on the critical path. Rather, it forces the ready jobs of subtasks on the critical path of each DAG to have the lowest priority among all ready jobs belonging to that DAG.  Thus, under the \cp policy, jobs of subtasks on the critical path of any DAG are executed as the last among all ready jobs of that DAG. Also note that for each dag-job,  any executing jobs of subtasks on the DAG's critical path may be preempted by newly released (and ready) jobs of subtasks which are not on the DAG's critical path. 

\begin{figure}[t]
\centering
\subfigure[\footnotesize An example DAG task $\tau_1$.]{\label{fig:T1task1}
\includegraphics[width=0.38 \columnwidth]{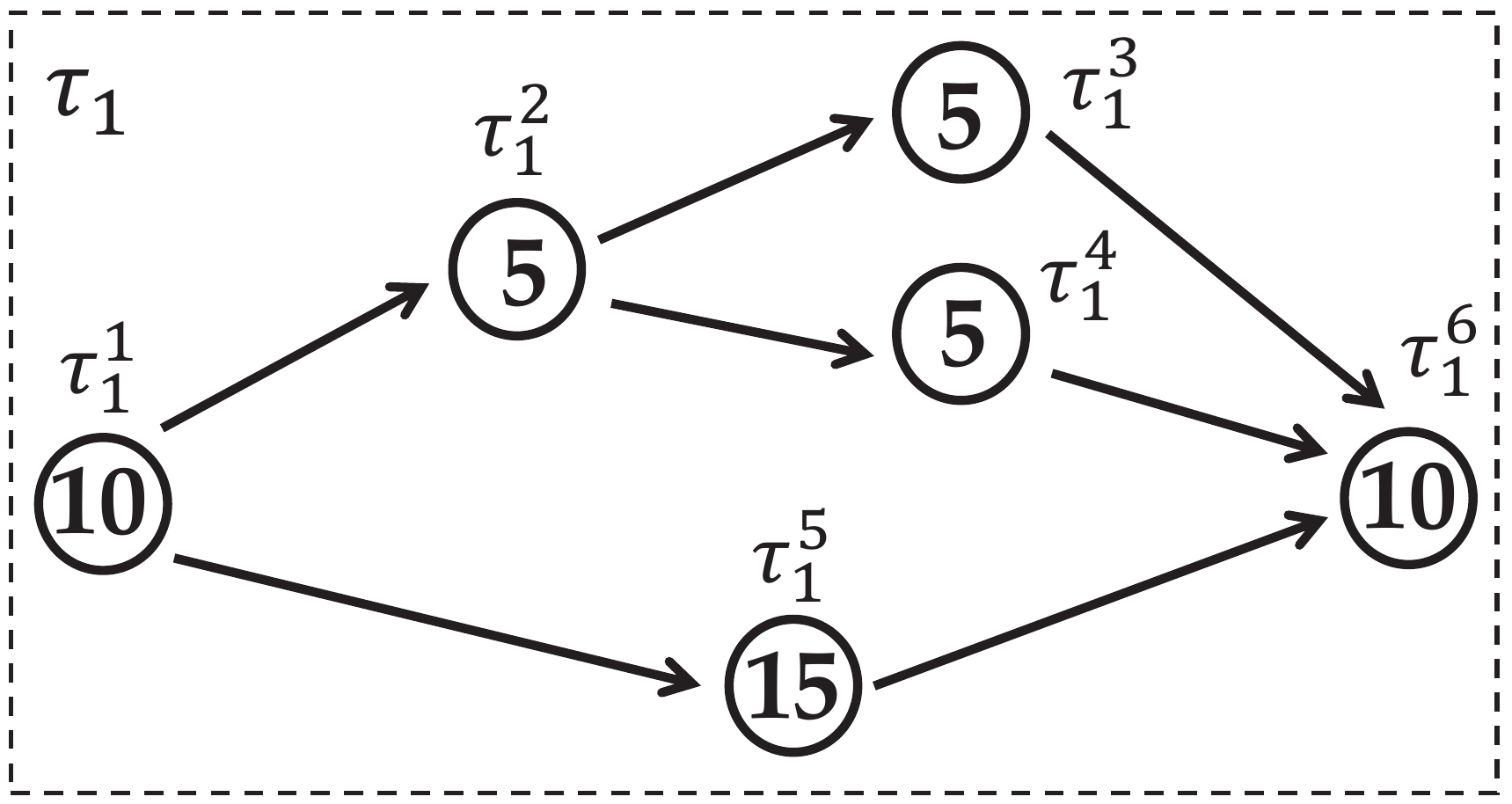}}
\hspace{+4pt}
\subfigure[The schedule of $\tau_1$.]{\label{fig:2processors}
\includegraphics[width=0.54\columnwidth]{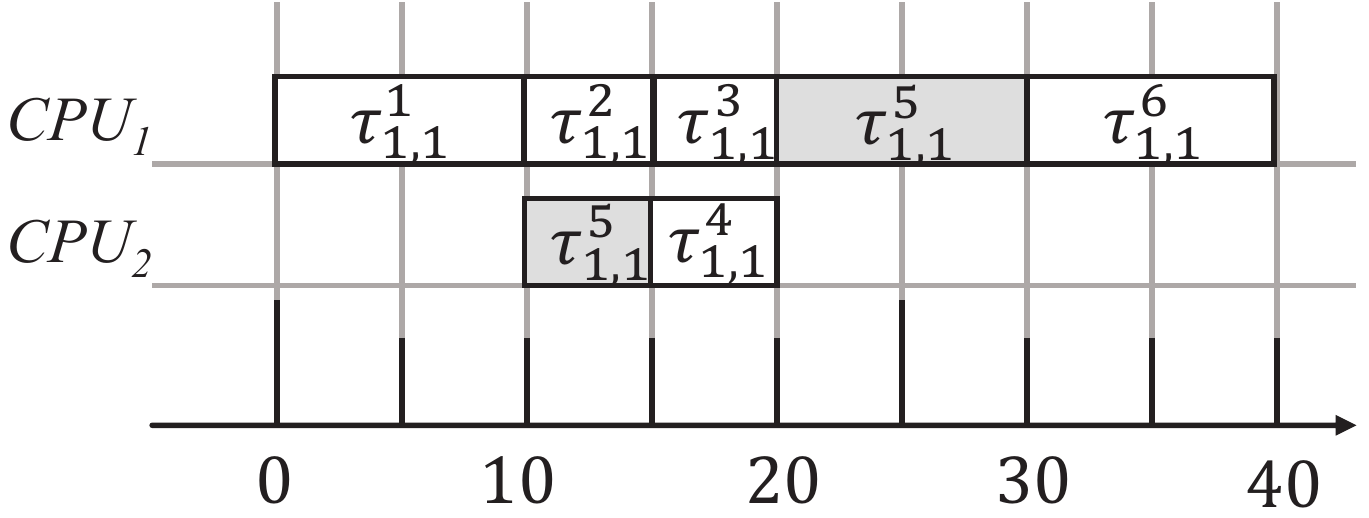}}
\caption{\footnotesize{Example illustrating CP-GEDF: $\tau_{1,1}^5$ is preempted by $\tau_{1,1}^3$ and $\tau_{1,1}^4$ due to the \cp policy.}}\label{Fig:lowestpriority}
\end{figure}


\begin{example}
Fig.~\ref{Fig:lowestpriority} shows that an example DAG $\tau_1$ scheduled on 2 processors under CP-GEDF. The DAG structure is shown in Fig.~\ref{fig:T1task1}, where vertices are labeled with the corresponding execution times. The critical path of $\tau_1$ is $\tau_1^1 \Rightarrow \tau_1^5 \Rightarrow\tau_1^6$. Fig.~\ref{fig:2processors} shows the schedule for the first dag-job of $\tau_1$. As seen in the figure, $\tau_{1,1}^5$ starts executing at time instant $10$. But according to the \cp policy, it gets preempted by $\tau_{1,1}^3$ and $\tau_{1,1}^4$ at time instant $15$. 
\end{example}

\subsection{Beneficial Properties of the \cp policy}
\label{sec:keyproperty}

We now derive beneficial properties of the CP-GEDF schedule. 

Let $\tau_{i,j}^{x_1}, \tau_{i,j}^{x_2},\dots, \tau_{i,j}^{x_k}$ denote the jobs of subtasks on the critical path for any dag-job $\tau_{i,j}$. Let $t_{i,j}^{x_u}$ denote the time instant when $\tau_{i,j}^{x_u}$ starts executing.

\begin{lemma}\label{lemma:readyjob}
 If we schedule a set $\tau = \{\tau_1, \dots, \tau_n\}$ of $n$ independent sporadic DAG tasks on $M$ identical processors under CP-GEDF, any job $\tau_{i,j}^{x_{u+1}}$ ($1\leq u \leq k-1$) belonging to the critical path of dag-job $\tau_{i,j}$ becomes ready when its predecessor job $\tau_{i,j}^{x_{u}}$ belonging to the critical path completes.
\end{lemma}

\begin{proof}
Let $\mathcal{S}$ denote the CP-GEDF schedule and $f_{i,j}^{x_u}$ denote the completion time of job $\tau_{i,j}^{x_u}$. To prove this lemma, it suffices to prove that the job $\tau_{i,j}^{x_{u}}$ belonging to the critical path is the last completed job among all predecessor jobs of $\tau_{i,j}^{x_{u+1}}$. We prove this new proof obligation by contradiction. Assume that $\tau_{i,j}^{x_{u+1}}$ has another predecessor job $\tau_{i,j}^{y_1}$ which completes later than $\tau_{i,j}^{x_{u}}$. Let $f_{i,j}^{y_1}$ denote the completion time of $\tau_{i,j}^{y_1}$.

Since under GEDF, dag-jobs released by different DAG tasks have distinct priorities and all jobs belonging to a dag-job inherit the same priority of the dag-job, the relative execution ordering of all jobs belonging to dag-job $\tau_{i,j}$ 			including $\tau_{i,j}^{x_{u}}$ and $\tau_{i,j}^{y_1}$  does not depend on other dag-jobs, but solely depends on the DAG structure of $\tau_i$ and the \cp policy.

\begin{figure}[t]
\centerline
{\includegraphics[width=0.45\textwidth]{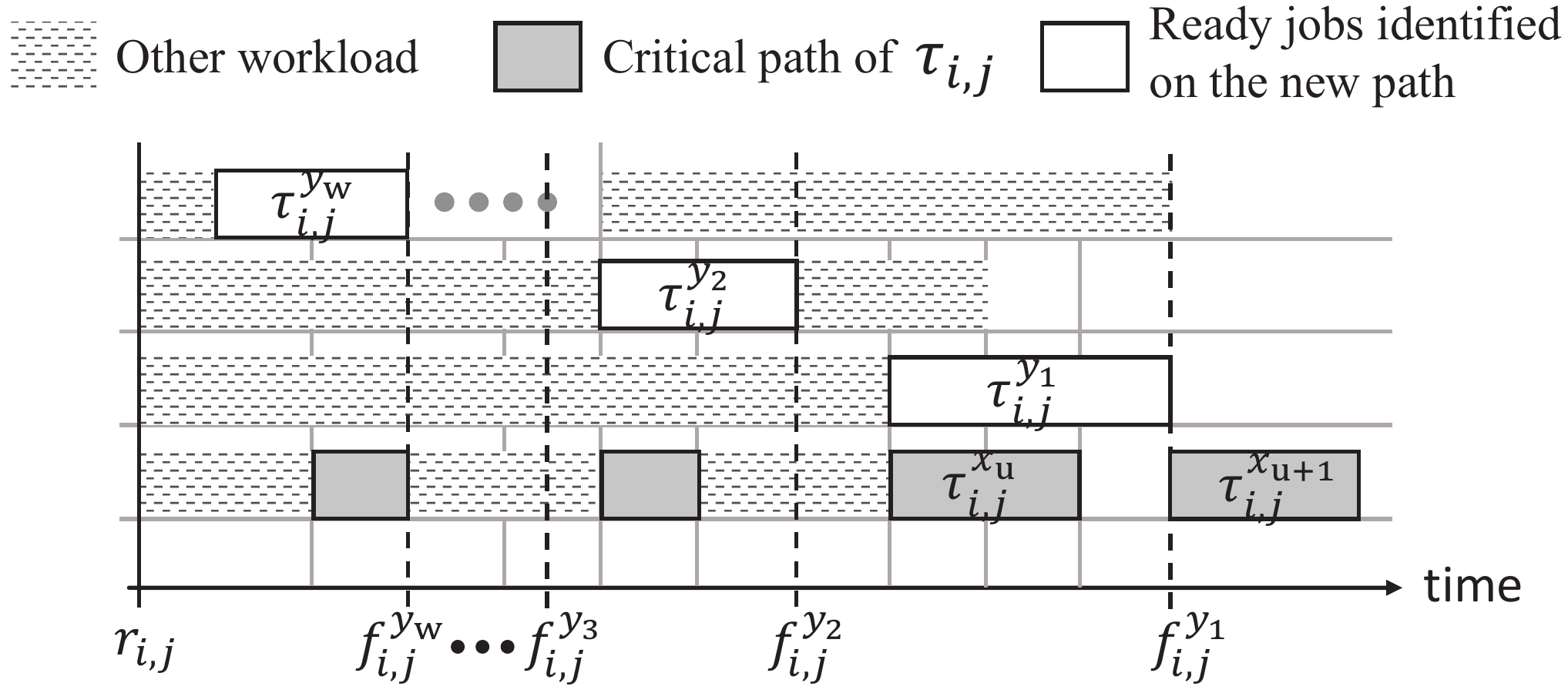}
}\caption{\footnotesize{Time intervals with respect to dag-job $\tau_{i,j}$.
}}\label{fig:illustrationlemmaone}
\end{figure}

We analyze the interval $[r_{i,j}, f_{i,j}^{y_1})$ by dividing it into $w \geq 1$ time intervals, denoted by $[f_{i,j}^{y_2}, f_{i,j}^{y_1}),  [f_{i,j}^{y_3}, f_{i,j}^{y_2}), \dots, [r_{i,j}, f_{i,j}^{y_w})$, ordered from right to left with respect to time, as illustrated in Fig.~\ref{fig:illustrationlemmaone}.
We identify these time intervals by moving from right to left with respect to time in the schedule $\mathcal{S}$ considering jobs belonging to the dag-job $\tau_{i,j}$.

We identify this interval set by finding a path in $\tau_i$ starting from a source subtask and ending at the subtask $\tau_i^{y_1}$.
Moving from the time instant $f_{i,j}^{y_1}$ to the left in $\mathcal{S}$, let $f_{i,j}^{y_2}$ denote the latest completion time among $\tau_{i,j}^{y_1}$'s predecessor jobs. Note that  $\tau_{i,j}^{y_1}$ becomes ready at $f_{i,j}^{y_2}$. We thus identify the first interval$[f_{i,j}^{y_2}, f_{i,j}^{y_1})$ in this set.
Moving from  $f_{i,j}^{y_2}$ to the left in $\mathcal{S}$, we apply this same process to identify the remaining intervals in this set, until we find the source subtask $\tau_{i}^{y_w}$ for this path, which is ready at the release time of the corresponding dag-job $\tau_{i,j}$ at $r_{i,j}$ and completes at $f_{i,j}^{y_w}$. Note that this path always exists because $\tau_{i,j}^{y_1}$ exists.

We now show a contradiction that there exists another path consisting of $\tau_{i}^{y_w},  \dots, \tau_{i}^{y_1}, \tau_{i}^{x_{u+1}}, \dots, \tau_{i}^{x_k}$ that is longer than the critical path of $\tau_i$. For each time interval $[f_{i,j}^{y_{q+1}}, f_{i,j}^{y_q})$ $(2\leq q\leq w-1)$, we know that $\tau_{i,j}^{y_q}$ is ready at the beginning of this interval $f_{i,j}^{y_{q+1}}$ according to the definition of the interval. Thus, each interval $[f_{i,j}^{y_{q+1}}, f_{i,j}^{y_q})$ (including $[r_{i,j}, f_{i,j}^{y_w})$) only consists of two kinds of subintervals: (\textit{i}) subintervals during which $\tau_{i,j}^{y_q}$ executes continuously, and (\textit{ii}) subintervals during which $\tau_{i,j}^{y_q}$ are not executing (i.e., being preempted by jobs belonging to other dag-jobs with higher priorities than $\tau_{i,j}$ or being delayed by jobs belonging to the same dag-job). During the first kind of subintervals, jobs from the subtasks on the critical path of $\tau_i$ may execute; but during the second kind of subintervals, jobs from the subtasks on the critical path of $\tau_i$ do not execute due to the \cp policy, for otherwise $\tau_{i,j}^{y_q}$ should have been executing during such subintervals. Thus, during each $[f_{i,j}^{y_{q+1}}, f_{i,j}^{y_q})$ ($(2\leq q\leq w-1)$), the workload executed due to $\tau_{i,j}^{y_q}$ is at least the workload due to jobs from the subtasks on the critical path of $\tau_i$. Moreover, within the last interval $[f_{i,j}^{y_{2}}, f_{i,j}^{y_1})$, since $\tau_{i,j}^{y_1}$ completes later than $\tau_{i,j}^{x_u}$, we know that the workload executed due to $\tau_{i,j}^{y_1}$ must be strictly greater than the workload due to jobs from the subtasks on the critical path of $\tau_i$. Therefore, the total workload executed within $[r_{i,j}, f_{i,j}^{y_{1}})$ due to the subtask set $\{\tau_{i}^{y_w},  \dots, \tau_{i}^{y_2}, \tau_{i}^{y_1}\}$ must be strictly greater than the total workload due to the subset of subtasks on the critical path of $\tau_i$ (this subset including subtasks on the partical critical path starting from the source subtask on the critical path and ending at $\tau_i^{x_u}$).

Clearly, we have identified another path in $\tau_i$, composed by $\tau_{i}^{y_w} \Rightarrow \tau_{i}^{y_{w-1}} \Rightarrow \dots \Rightarrow \tau_{i}^{y_2} \Rightarrow \tau_{i}^{y_1} \Rightarrow \tau_i^{x_{u+1}} \Rightarrow \Gamma$, where $\Gamma$ represents the remaining subtasks following $\tau_i^{x_{u+1}}$ on the critical path of $\tau_i$. And this path is longer than the critical path of $\tau_i$. A contradiction is reached.
\end{proof}

\begin{lemma}\label{lemma:lastjob}
 If we schedule a set $\tau = \{\tau_1, \dots, \tau_n\}$ of $n$ independent sporadic DAG tasks on $M$ identical processors under CP-GEDF, a dag-job $\tau_{i,j}$ of $\tau_i$ completes its execution when the  job $\tau_{i,j}^{x_k}$ of the last subtask $\tau_{i}^{x_k}$ on the critical path complete its execution.
\end{lemma}
\begin{proof}
This lemma can be proved in the same manner as lemma~\ref{lemma:readyjob}, by reaching a contradiction where an identified path in $\tau_i$ is longer than the critical path of $\tau_i$. For completeness, we put the detailed proof in the appendix.
\end{proof}

\section{Schedulability Analysis}\label{sec:schedulability}

We now present our schedulability analysis for DAG task systems scheduled under CP-GEDF. Our approach is fundamentally based on the window-based reasoning framework which was first proposed by Baker~\cite{baker2003} and has been used extensively to analyze the ordinary sporadic task model~\cite{bertogna2009schedulability, baruah2007techniques, bertogna2005improved}. Due to the complex DAG structure and intra-DAG precedence constraints, the original window-based reasoning could not precisely characterize the workload distribution in an analysis window. We first enhance this framework by developing an analysis technique based on a novel concept of executing/non-executing critical path interval.

\subsection{Executing/non-executing intervals}\label{sec:executingnonexectuing}

We first present the formal definition of an executing/non-executing critical path interval.


\begin{definition}\label{def:executinginterval}
\textbf{(Executing/non-executing critical path interval)} A time interval $[t_1, t_2)$ is an executing critical path interval for $\tau_i$ if the following three conditions hold: (\text{i}) a DAG job $\tau_{i,j} $ of $\tau_i$ is released at or before $t_1$; (\text{ii}) $\tau_{i,j}$ completes its execution no earlier than $t_2$; (\text{iii}) jobs of subtasks on $\tau_{i}$'s critical path are executing continuously throughout this interval. Otherwise, if conditions (\text{i}) and (\text{ii}) hold, but no job of any subtask on $\tau_i$'s critical path executes at any time instant within $[t_1, t_2)$, then $[t_1, t_2)$ is a non-executing critical path interval for $\tau_i$.
\end{definition}

\begin{definition}
\label{def:busyinterval}
\textbf{(Busy interval)} a time instant $t$ is a busy instant if all $M$ processors execute jobs at $t$. A time interval $[t_1, t_2)$ is busy if every time instant $t$ in $[t_1, t_2)$ is busy.
\end{definition}

\begin{figure}
\centerline
{\includegraphics[width=0.48\textwidth]{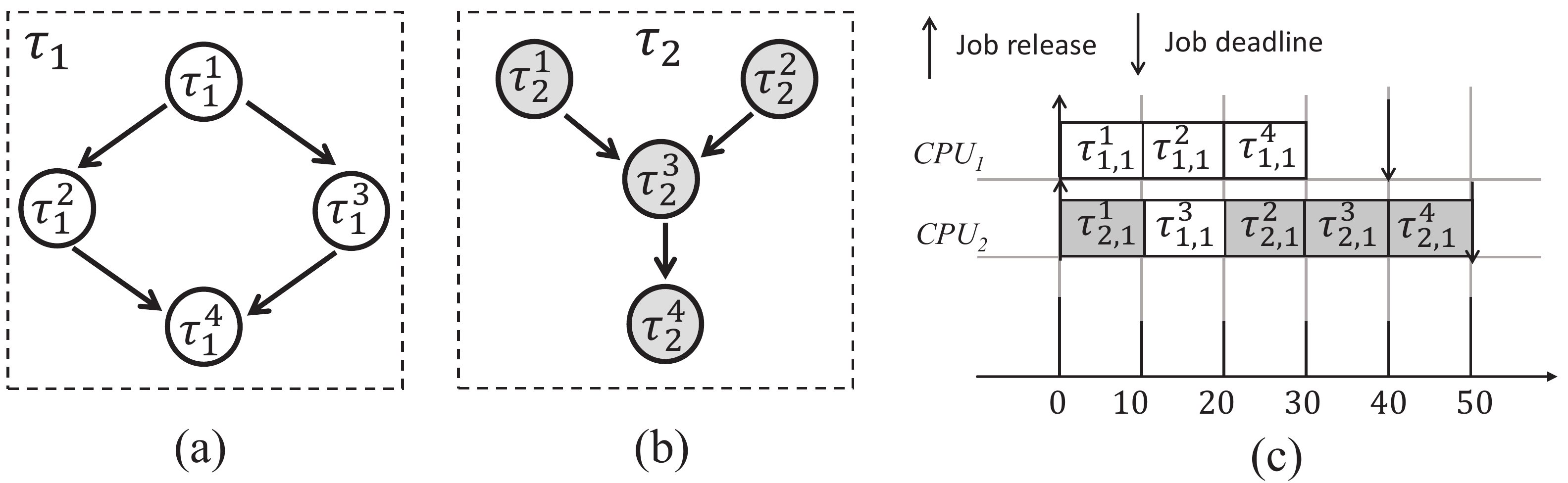}
}\caption{\footnotesize{Example illustrating the executing/non-executing critical path interval and the busy interval.}}\label{fig:exenonexe}
\end{figure}

\begin{example}
Fig.~\ref{fig:exenonexe} shows that two DAG tasks are scheduled on two processors under CP-GEDF. The DAG structures for $\tau_1$ and $\tau_2$ are shown in Fig.~\ref{fig:exenonexe}(a) and Fig.~\ref{fig:exenonexe}(b), where the execution cost for each subtask is 10 time units. Fig.~\ref{fig:exenonexe}(c) shows the schedule for the first dag-jobs of $\tau_1$ and $\tau_2$. Since the periods of $\tau_1$ and $\tau_2$ are 40 time units and 50 time units, respectively, and both $\tau_1$ and $\tau_2$ release the first dag-jobs at time 0, jobs belonging to the first dag-job of $\tau_1$ have higher priorities than those of $\tau_2$. According to the DAG structure, the critical path of $\tau_2$ is $\tau_2^1 \Rightarrow \tau_2^3 \Rightarrow\tau_2^4$.  As seen in the Fig.~\ref{fig:exenonexe}(c), $\tau_{2,1}^1$ starts executing at
time $0$, but $\tau_{2,1}^3$ is preempted during $[10, 30)$ according to GEDF-CP. In this example, $[0, 10)$ and $[30, 50)$ are executing critical path intervals for $\tau_2$, and $[10, 30)$ is a non-executing critical path interval for $\tau_2$. $[0, 30)$ is clearly a busy interval. 
\end{example}


\begin{lemma}\label{lemma:busy}
In the CP-GEDF schedule for a DAG task system $\tau$, if $[t_1, t_2)$ is a non-executing critical path interval for $\tau_i$, then $[t_1, t_2)$ is a busy interval.
\end{lemma}
\begin{proof}
Let $\tau_{i,j}$ denote the dag-job that is released before $t_1$ and does not complete before $t_2$. We prove this lemma by contradiction. Suppose $[t_1, t_2)$ is a non-executing critical path interval for $\tau_i$, and $[t_1, t_2)$ is not a busy interval. By Def.~\ref{def:busyinterval}, at least one processor is idle during $[t_1, t_2)$. By Lemma~\ref{lemma:lastjob}, some jobs of subtasks on $\tau_i$'s critical path have not completed before $t_2$. Let $\beta$ denote the set of such jobs. Note that $\beta$ must exist because at least the job of the last subtask on $\tau_i$'s critical path has not completed before $t_2$. According to Lemma~\ref{lemma:readyjob}, there must be a job in $\beta$ , denoted by $J$, that is ready at $t_1$. Thus, $J$ must be executing on the idle processor during $[t_1, t_2)$. However, by Def.~\ref{def:executinginterval}, no job of any subtask on $\tau_i$'s critical path shall execute at any time instant within $[t_1, t_2)$. A contradiction is thus reached.
\end{proof}

\subsection{A Necessary Condition for Deadline Misses}
\label{sec:necessarycondition}

We focus on analyzing what happens when a deadline is missed given any DAG task system $\tau$ scheduled under CP-GEDF on $M$ identical processors. Let $t_d$ denote the first time instant in any such schedule $\mathcal{S}$ at which a deadline is missed. Let dag-job $\tau_{h,l}$ be the one that misses its deadline $d_{h,l}$ at $t_d$, which is released by task $\tau_h$ at $r_{h,l}$. Note that dag-jobs with deadlines later than $t_d$ do not affect the scheduling of dag-jobs with deadlines no later than $t_d$. Thus, we remove every dag-job with a deadline later than $t_d$ from $\mathcal{S}$.

\begin{definition}
\textbf{(Problem task, problem job, problem window)} based on the above discussion, $\tau_h$ is a problem task, $\tau_{h,l}$ is a problem dag-job, and the time interval $[r_{h,l}, d_{h,l})$ is a problem window.
\end{definition}

\begin{definition}\label{def:workload}
\textbf{(Workload)} the workload $W$ within a time interval $[t, t+\Delta)$ is the total amount of computation executed within this time interval in the schedule $\mathcal{S}$.
\end{definition}

\begin{definition}
\textbf{(Average workload)} the average workload within a time interval $[t, t+\Delta)$ is $\frac{W}{\Delta}$, where $W$ is the workload within this interval.
\end{definition}

The following lower bound on the average workload of a problem window can be observed in the schedule $\mathcal{S}$, which is a necessary condition for $\tau_{h,l}$ to miss its deadline.

\begin{lemma}
\label{lemma:misscondition}
Since $\tau_{h,l}$ misses its deadline at $t_d$, the sum of the lengths of all non-executing critical path intervals for $\tau_h$ within $[r_{h,l}, t_d)$ must exceed $p_h-\mathbf{len}(\mathbf{Cpath}_h)$.
\end{lemma}

\begin{proof}
According to lemma~\ref{lemma:lastjob}, a dag-job completes when the job of the last subtask on its critical path completes. Since $\tau_{h,l}$ has not completed by $t_d$, the job of the last subtask on $\tau_{h}$'s critical path  has not completed by $t_d$. We can divide $[r_{h,l}, t_d)$ into subintervals including either executing or non-executing critical path intervals for $\tau_{h}$ by Def.~\ref{def:executinginterval}. Since the jobs of subtasks on the critical path must execute sequentially, this lemma immediately follows. 
\end{proof}

\begin{figure}
\centerline
{\includegraphics[width=0.48\textwidth]{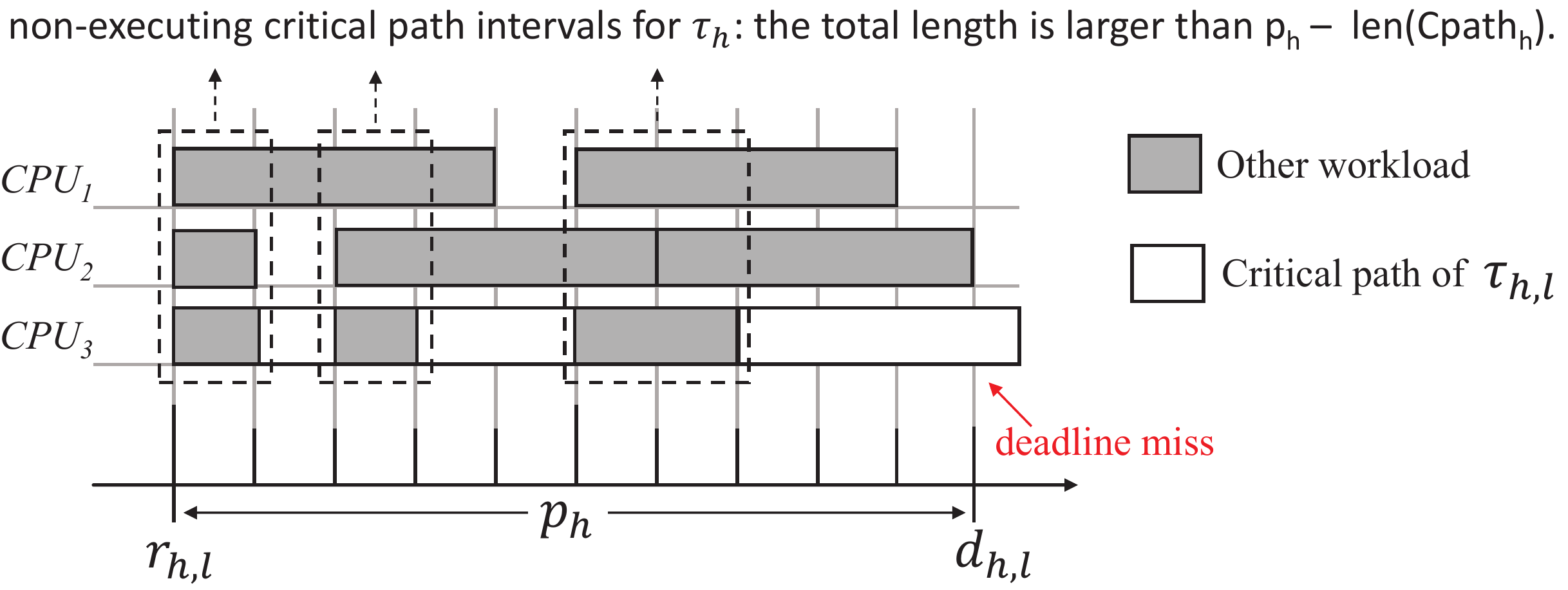}
}\caption{\footnotesize{Example illustrating Lemma~\ref{lemma:misscondition}.}}\label{fig:exampleworkload}
\end{figure}


\begin{example}
Fig.~\ref{fig:exampleworkload} shows a general example to illustrate the necessary condition given in Lemma~\ref{lemma:misscondition}. Since $\tau_{h,l}$ misses its deadline at $t_d$, the corresponding released jobs of subtasks on the critical path of $\tau_{h}$ have not completed by $t_d$ according to lemma~\ref{lemma:lastjob}. As seen in the figure, the sum of the lengths of all non-executing critical path intervals for $\tau_h$ within $[r_{h,l}, t_d)$ must exceed $p_h-\mathbf{len}(\mathbf{Cpath}_h)$.
\end{example}

\begin{lemma}\label{lemma:lowerbound}
Since $\tau_{h,l}$ misses its deadline at $t_d$,
\begin{equation}\label{eq:neccondition}
\frac{W}{p_h} > M \times (1 - \sigma_h) + \sigma_h
\end{equation}
holds where $\sigma_h = \frac{\mathbf{len}(\mathbf{Cpath}_h)}{p_h}$ and $W$ is the workload within $[r_{h,l}, t_d)$.
\end{lemma}

\begin{proof}
Let $\alpha$ denote the total length of all executing critical path intervals for $\tau_{h}$ within $[r_{h,l}, d_{h,l})$. Thus, by Def.~\ref{def:executinginterval}, the total length of all non-executing critical path intervals for $\tau_{h}$ within $[r_{h,l}, d_{h,l})$ is given by $p_h - \alpha$. According to lemma~\ref{lemma:busy}, all processors are busy during any non-executing critical path interval for $\tau_h$. Thus, we have
\begin{equation}
\begin{split}
W &\geq M\times(p_h - \alpha) + \alpha\\
    &> M\times p_h - (M - 1)\times\mathbf{len}(\mathbf{Cpath}_h).
\end{split}
\end{equation}

The second $>$ in the above equation holds because $\alpha < \mathbf{len}(\mathbf{Cpath}_h)$ as $\tau_{h,l}$ misses its deadline at $t_d$. If we divide both sides of the inequality by $p_h$, the lemma follows.
\end{proof}

Note that in the above lemma, $\frac{W}{p_h}$ represents the average workload within our problem window $[r_{h,l}, d_{h,l})$. Eq.~\ref{eq:neccondition} shows a necessary average workload condition for the deadline miss to happen.



\subsection{Window-based Analysis}\label{sec:window}

We now present a window-based analysis for upper-bounding the workload $W$ within the problem window, which allows us to derive a schedulability test using Lemma~\ref{lemma:lowerbound}.

\begin{figure}
\centerline
{\includegraphics[width=0.48\textwidth]{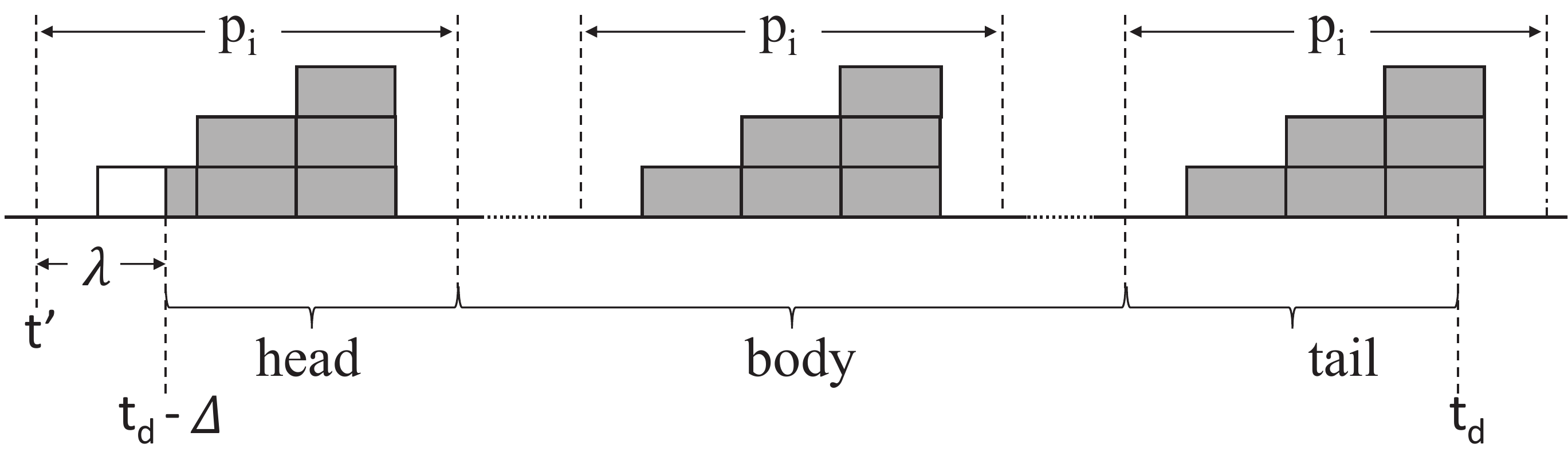}
}\caption{\footnotesize{The head, body, and tail sub-windows. }}\label{fig:windowofinterest}
\end{figure}


For any DAG $\tau_i$ that may execute in a problem window of interest, we divide this window into three sub-windows including the head, the body, and the tail sub-window. The workload contributed by $\tau_i$ in this window, denoted by $W_i$, is the sum of the workload contributed by $\tau_i$ within these three sub-windows. We seek to obtain a total upper bound on $W_i$ by seperately upper-bounding the workload contribution of $\tau_i$ in each sub-window. This concept is illustrated in Fig.~\ref{fig:windowofinterest}. The head of a window $[t_d - \Delta, t_d)$ is defined to be interval $[t_d - \Delta, t_d - \Delta+min\{\Delta, p_i - \lambda\})$, if there exists a dag-job of DAG task $\tau_i$ that is released at time $t' = t_d - \Delta - \lambda$ and $0< \lambda < p_i$. Such a dag-job is called the carry-in dag-job of $\tau_i$ w.r.t. this problem window (formally defined in Def.~\ref{def:carry-injob}). Note that if such a carry-in dag-job of $\tau_i$ does not exist, $\tau_i$ does not contribute any workload within the head sub-window.
Besides the head sub-window, the tail sub-window exists if $\tau_i$ releases a dag-job before $t_d$ which has a deadline later than $t_d$. The remaining interval within $[t_d - \Delta, t_d)$ is defined to be the body sub-window.



Based on the above definition, it is straightforward to upper bound the workload contributed by $\tau_i$ during the body sub-window through upper-bounding the number of dag-jobs released by $\tau_i$ in the body sub-window. Moreover, $\tau_i$ does not contribute any workload during the tail sub-window under CP-GEDF, since we already removed dag-jobs with deadlines later than $t_d$ from $\mathcal{S}$.
The challenge mainly lies in upper-bounding the workload contributed by $\tau_i$ within the head sub-window.
In the rest of  this section, we first derive a workload upper bound within the head sub-window in Lemmas~\ref{carryinbound} and~\ref{lemma:upperboundon}, and then derive the overall workload $W_i$ within the entire problem window in Lemma~\ref{lemma:upperbound}.

To derive a tighter upper bound on the workload with the head window, we apply the same window extension technique as first proposed in \cite{baker2003}. We extend the original problem window $[r_{h,l}, t_d)$ to find the \textit{maximal $\sigma_h$-busy window: $[t_d-\Delta, t_d)$}, where $t_d-\Delta \leq r_{h,l}$, which is defined as follows:


\begin{definition}\label{def:busy}
\textbf{($\sigma_h$-busy)}  A time interval is $\sigma_h$-busy if its average workload  is at least $M\times(1-\sigma_h) + \sigma_h$. Note that the original problem window $[r_{h,l}, t_d)$ is $\sigma_h$-busy according to lemma~\ref{lemma:lowerbound}.
\end{definition}

\begin{definition}\label{def:maximal}
\textbf{(Maximal $\sigma_h$-busy window)} The maximal $\sigma_h$-busy window, denoted by $[t_d-\Delta, t_d)$ ($t_d-\Delta \leq r_{h,l}$), is a downward extension of the $\sigma_h$-busy interval $[r_{h,l}, t_d)$, which has no longer downward extensions that are $\sigma_h$-busy. Note that such a maximal $\sigma_h$-busy window exists as at least $[r_{h,l}, t_d)$ serves as one.
\end{definition}

Intuitively, the definition of the maximal $\sigma_h$-busy window describes the workload distribution on the schedule $\mathcal{S}$: the average workload within $[t_d-\Delta, t_d)$ is at least $M\times(1-\sigma_h) + \sigma_h$; while the average workload within any interval $[t, t_d - \Delta)$, where $0 \leq t <t_d - \Delta$, must be smaller than $M\times(1-\sigma_h) + \sigma_h$. This property will be used to upper-bound the carry-in workload in lemma~\ref{lemma:upperboundon}.

\begin{lemma}\label{lemma:existence}
$\tau_h$ has a unique maximal $\sigma_h$-busy window $[t_d-\Delta, t_d)$, for $\sigma_h$.
\end{lemma}
\begin{proof}
By lemma~\ref{lemma:lowerbound}, the interval $[r_{h,l}, t_d)$ is $\sigma_h$-busy. Thus, by Def.~\ref{def:maximal}, $\tau_h$ has a unique maximal $\sigma_h$-busy window $[t_d-\Delta, t_d)$, where $t_d-\Delta \leq r_{h,l}$. 
\end{proof}

We will now focus on the maximal $\sigma_h$-busy window $[t_d-\Delta, t_d)$, to analyze schedulability. 

\begin{definition}\label{def:carry-injob}
\textbf{(Carry-in dag-job)} The carry-in dag-job of $\tau_i$ for the time window $[t_d-\Delta, t_d)$ is the last dag-job of task $\tau_i$ released before time instant $t_d-\Delta$.
\end{definition}


\begin{definition}
\textbf{(Carry-in workload)} The carry-in workload of $\tau_i$ at time instant $t$, denoted by $\varepsilon_i$,  is the remaining workload of the carry-in dag-job of task $\tau_i$ at time instant $t$.
\end{definition}

As shown in Fig.~\ref{fig:windowofinterest}, let $t'$ denote the release time of the carry-in dag-job of $\tau_i$, i.e., $t' = t - \lambda$, where $\lambda$ is the offset of the release time from the beginning of the window. In the following two lemmas, we upper bound the carry-in workload of $\tau_i$ in the maximal $\sigma_h$-busy window $[t_d - \Delta, t_d)$.




\begin{figure}
\centerline
{\includegraphics[width=0.48\textwidth]{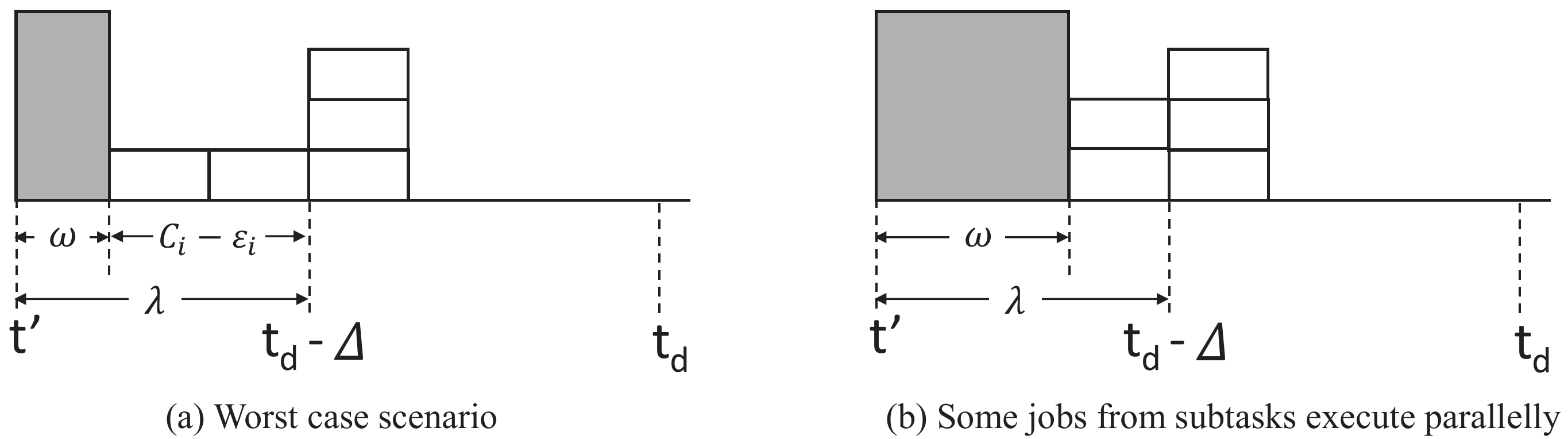}
}\caption{\footnotesize{Carry-in workload depends on the competing workload in $[t', t_d-\Delta)$. }}\label{fig:worstcasescenario}
\end{figure}

\begin{lemma}\label{carryinbound}
The average workload of interval $[t', t_d-\Delta)$ is lower bounded by $(M-1)\times\frac{(\lambda -  C_i +\varepsilon_i) }{\lambda}+ 1$.
\end{lemma}

\begin{proof}
This lemma seeks to prove a lower bound on the average workload of interval $[t', t_d-\Delta)$, when $\tau_i$'s carry-in workload is $\varepsilon_i$ within $[ t_d-\Delta, t_d)$.

Intuitively, there are three types of workloads within this interval: (1) the workload due to the dag-job of $\tau_i$ (2) the workload due to tasks other than $\tau_i$ which preempts $\tau_i$, and (3) the workload that executes with $\tau_i$'s dag-job in parallel. Since the first type of workload is given by $C_i - \varepsilon_i$,
the total amount of the first two types of the workload in $[t', t_d-\Delta)$ is given by $M \times \omega + C_i - \varepsilon_i$, where $\omega$ is the total length of all the intervals in $[t', t_d-\Delta)$ where all $M$ processors are busy (because $\tau_i$ gets preempted in such intervals). Since $\tau_i$ is a DAG task, its DAG structure can impact the value of $\omega$. Fig.~\ref{fig:worstcasescenario}~(a) shows the case when $\omega$ reaches its smallest value. That is, when jobs of subtasks of $\tau_i$ execute sequentially during $[t', t_d-\Delta)$, $\omega$ achieves its smallest value. Otherwise, if some jobs of subtasks of $\tau_i$ execute in parallel as shown in Fig.~\ref{fig:worstcasescenario}~(b), $\omega$ becomes larger. This is because $\lambda$ is a fixed value and $C_i - \varepsilon_i$ is the amount of workload which gets executed due to the dag-job of $\tau_i$ in $[t', t_d-\Delta)$.  Thus, the smallest $\omega$ is given by $\omega = \lambda - (C_i - \varepsilon_i)$. By ignoring the third type of workload in $[t', t_d-\Delta)$, we can lower-bound the average workload within interval $[t', t_d-\Delta)$ by:
\begin{equation}
\begin{split}
\frac{M\times \omega + (C_i - \varepsilon_i)}{\lambda} &= \frac{M\times \omega + (\lambda - \omega)}{\lambda} \\
 &= (M-1)\times \frac{\omega}{\lambda} + 1 \\
&= (M-1)\times \frac{\lambda - C_i + \varepsilon_i}{\lambda} + 1.
\end{split}
\end{equation}
\end{proof}


\begin{lemma}\label{lemma:upperboundon}
The carry-in workload $\varepsilon_i$ of $\tau_i$ within the maximal $\sigma_h$-busy window $[t_d-\Delta, t_d)$ is upper bounded  by $C_i - \sigma_h\times\lambda$.
\end{lemma}

\begin{proof}
Since $[t_d-\Delta, t_d)$  is the maximal $\sigma_h$-busy window, according to Def.~\ref{def:maximal}, the average workload of $[t', t_d-\Delta)$ is smaller than $M\times(1-\sigma_h) + \sigma_h$, for otherwise $[t', t_d)$ should have been the maximal $\sigma_h$-busy window. According to lemma~\ref{carryinbound}, we have
\begin{equation}
(M-1)\times \frac{\lambda - C_i + \varepsilon_i}{\lambda} + 1 < M\times(1-\sigma_h) + \sigma_h.
\end{equation}
Thus, we have $0 \leq \varepsilon_i < C_i - \sigma_h\times\lambda$. 
\end{proof}

As discussed earlier, since any dag-job with a deadline later than $t_d$ is removed from the schedule, $\tau_i$ does not contribute any workload in the tail sub-window. Thus, the contribution due to $\tau_i$ in the entire problem window $[t_d - \Delta, t_d)$ only depends on its carry-in workload contributed by $\tau_i$ in the head sub-window and its workload within the body sub-window. The following Lemma~\ref{lemma:boundWi} gives an upper bound on $W_i$ by summing up the workload within the head sub-window and the body sub-window, following the same reasoning provided by lemma~10 in~\cite{baker2003}, which is used to upper bound the workload due to an ordinary sporadic task in a problem window. The same reasoning can be applied herein because intuitively, the workload due to a DAG task $\tau_i$ in the body sub-window depends solely on the number of dag-jobs released by $\tau_i$ within this sub-window, which is exactly the same as the ordinary sporadic task case. That is, the number of such released dag-jobs (or ordinary sporadic jobs) within the body sub-window is mainly constrained by the period $p_i$. We put the detailed proof in the appendix for completeness.


\begin{lemma}\label{lemma:boundWi}
The workload $W_i$ of $\tau_i$ on $\mathcal{S}$ during $[t_d-\Delta, t_d)$ is no greater than
\begin{equation}
    W_i=
    \begin{cases}
      \lfloor\frac{\Delta}{p_i}\rfloor\times C_i + \max\{0, C_i - \sigma_h\times\lambda\}, & \text{if}\ \Delta \geq p_i \\
      \max\{0, C_i - \sigma_h\times\lambda\}, & \text{otherwise}
    \end{cases}
\end{equation}
\noindent where $\lambda = (\lfloor\frac{\Delta}{p_i}\rfloor + 1)\times p_i - \Delta$.
\end{lemma}


The following lemma is to upper-bound the average workload $\frac{W_i}{\Delta}$ due to $\tau_i$ within $[t_d-\Delta, t_d)$. Similar to Lemma~\ref{lemma:boundWi}, this process is essentially the same for the DAG and the oridinary sporadic task models.

\begin{lemma}\label{lemma:upperbound}
For the maximal $\sigma_h$-busy window $[t_d-\Delta, t_d)$, the average workload $\frac{W_i}{\Delta}$ due to $\tau_i$ is at most $\eta_i$, where
\begin{equation}\label{eq:eta}
    \eta_i=
    \begin{cases}
      u_i, & \text{if}\ \sigma_h \geq  u_i\\
      u_i + \frac{C_i - \sigma_h \times p_i}{p_h}, & \text{if}\ \sigma_h <  u_i.
    \end{cases}
\end{equation}
\end{lemma}


\subsection{A Utilization-based Schedulability Test}\label{sec:test}


We now use Lemma~\ref{lemma:upperbound}  combined with Lemma~\ref{lemma:lowerbound}, which is the necessary condition for deadline misses given in Sec.~\ref{sec:necessarycondition}, to derive a utilization-based schedulability test.


\begin{theorem}\label{Theorem:test}
A set $\tau = \{\tau_1, \dots, \tau_n\}$ of $n$ independent sporadic DAG tasks is schedulable on $M$ identical processors under CP-GEDF, if, for every task $\tau_k$,
\begin{equation}\label{eq:finaltest}
\sum_{i=1}^n \eta_i \leq M - (M - 1) \times\sigma_k.
\end{equation}
\noindent where $\eta_i$ is defined in Eq.~\ref{eq:eta} and $\sigma_k$ is defined in Def.~\ref{definition:cputilization}.
\end{theorem}

\begin{proof}
We prove this lemma by contradiction. Suppose some deadline misses occur in the CP-GEDF schedule.
Let $\tau_h$ be the first task to miss a deadline at $t_d$ and $[t_d - \Delta, t_d)$ be the maximal $\sigma_h$-busy window with respect to $\tau_h$. The existence of $[t_d - \Delta, t_d)$ is guaranteed by Lemma~\ref{lemma:lowerbound} and Lemma~\ref{lemma:existence}. Since $[t_d - \Delta, t_d)$ is $\sigma_h$-busy, by Def.~\ref{def:maximal}, we have $\frac{W}{\Delta} > M\times(1-\sigma_h) + \sigma_h$. By lemma~\ref{lemma:upperbound}, we have $\frac{W_i}{\Delta}\leq \eta_i$, for $1\leq i \leq n$.
Since the workload of $\tau_i$ during $[t_d-\Delta, t_d)$ is upper bounded by $W_i$,  $\sum_{i=1}^n W_i \geq W$ by the definition of $W$ given in Def.~\ref{def:workload}. Thus, we have
\begin{equation}
\sum_{i=1}^n \eta_i \geq \sum_{i=1}^n \frac{W_i}{\Delta} \geq \frac{W}{\Delta} > M\times(1 - \sigma_h) + \sigma_h.
\end{equation}
 This clearly contradicts Eq.~\ref{eq:finaltest}.
\end{proof}

The above schedulability test in Theorem~\ref{Theorem:test} must be checked individually for each task, thus having a time complexity of $\mathcal{O}(n)$. This test can be viewed as the DAG version of the test designed for the ordinary sporadic task model given in~\cite{baker2003}, if  replacing $\sigma_k$ by $\tau_k$'s utilization.


%
%
%

\begin{figure*}
\centering
{\includegraphics[width=0.6\textwidth]{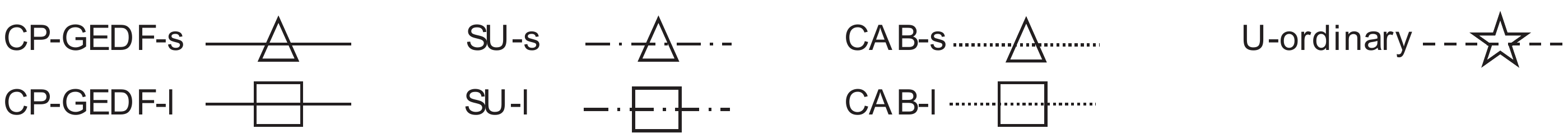}}\vspace{-1pt}\\
\subfigure[M=8, light per-task utilization.]{\label{fig:exp1}
\includegraphics[width=0.66\columnwidth]{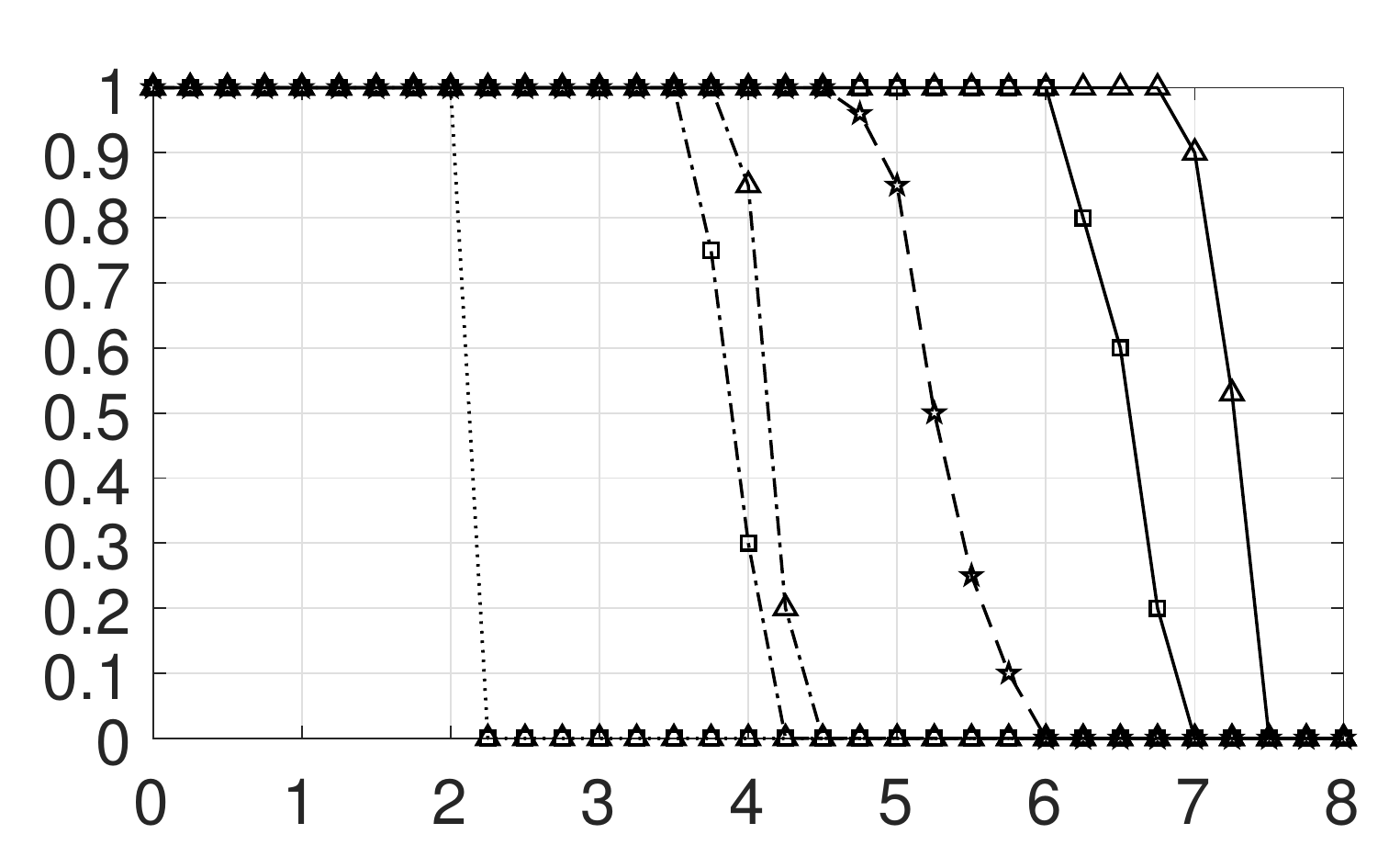}}
\subfigure[M=8, medium per-task utilization.]{\label{fig:exp3}
\includegraphics[width=0.66\columnwidth]{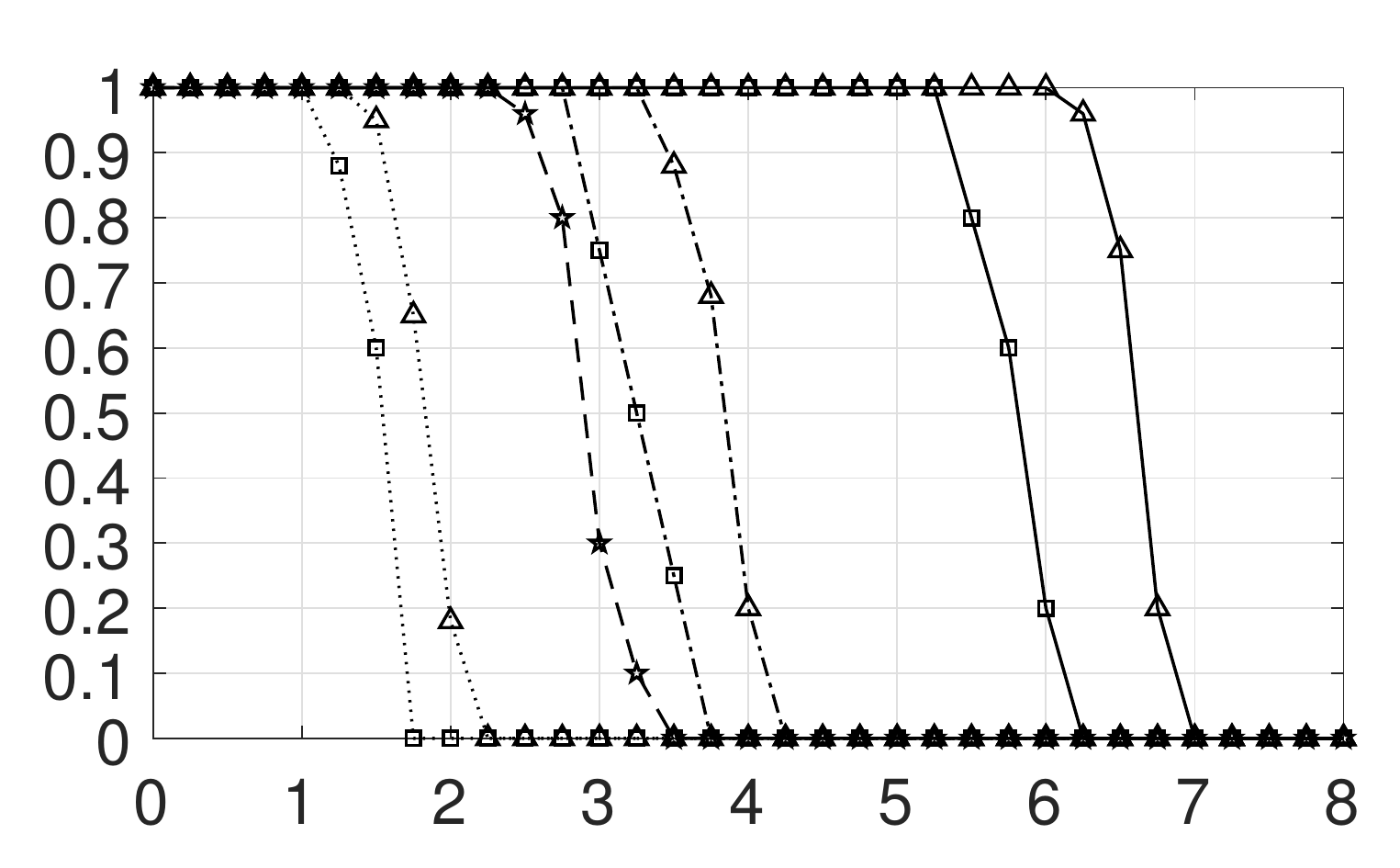}}
\subfigure[M=8, heavy per-task utilization.]{\label{fig:exp5}
\includegraphics[width=0.66\columnwidth]{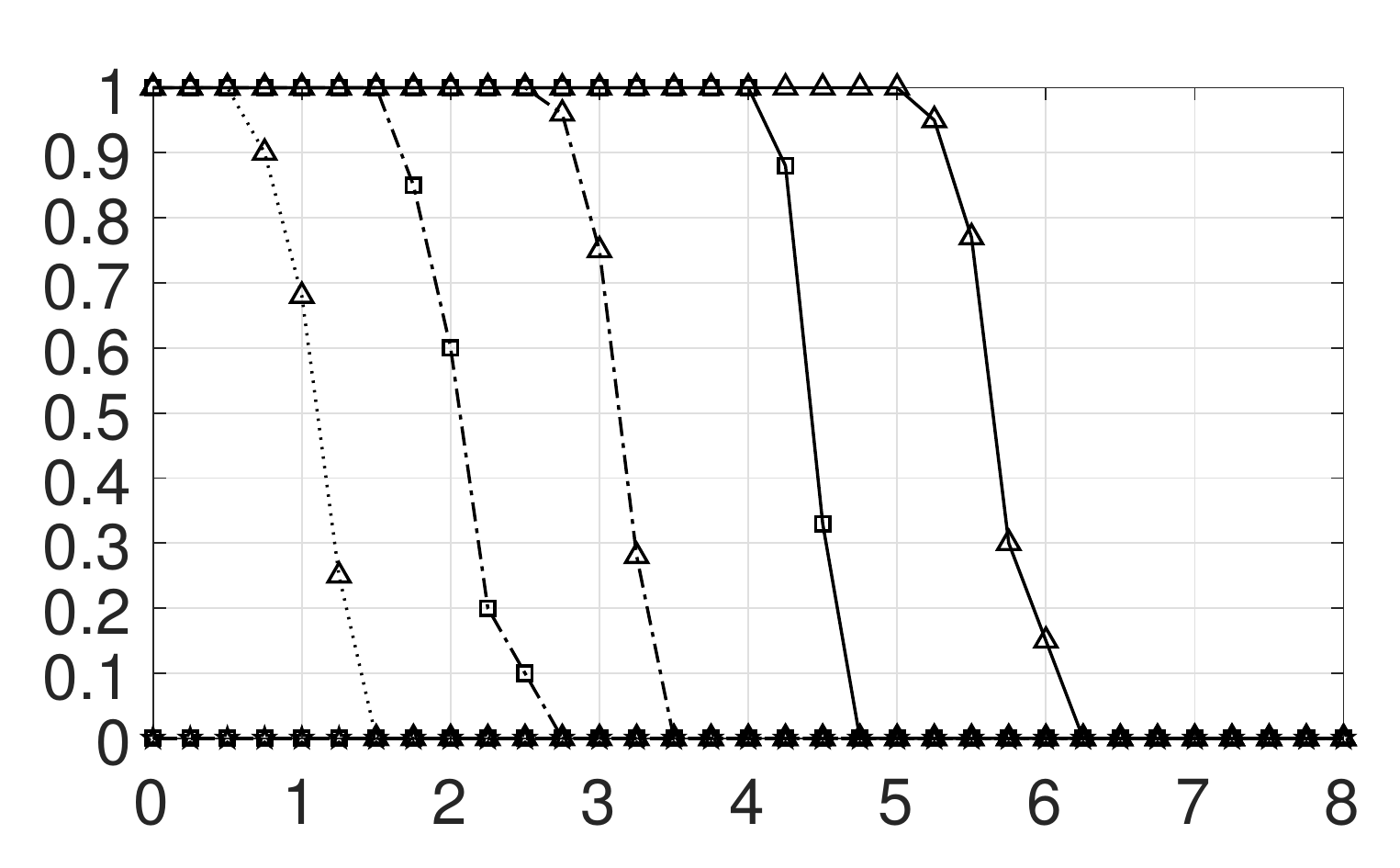}}
\subfigure[M=16, light per-task utilization.]{\label{fig:exp2}
\includegraphics[width=0.66\columnwidth]{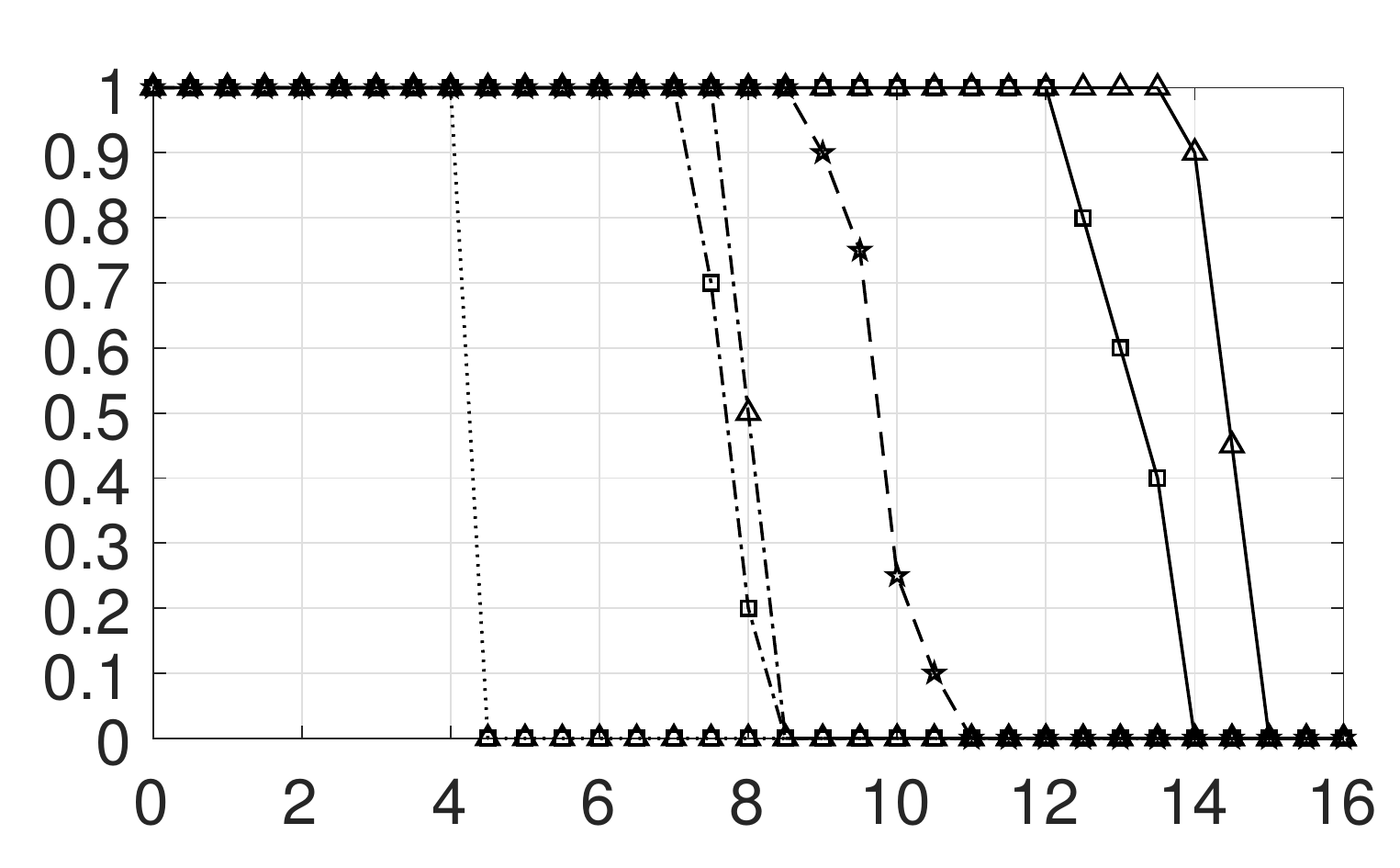}}
\subfigure[M=16, medium per-task utilization.]{\label{fig:exp4}
\includegraphics[width=0.66\columnwidth]{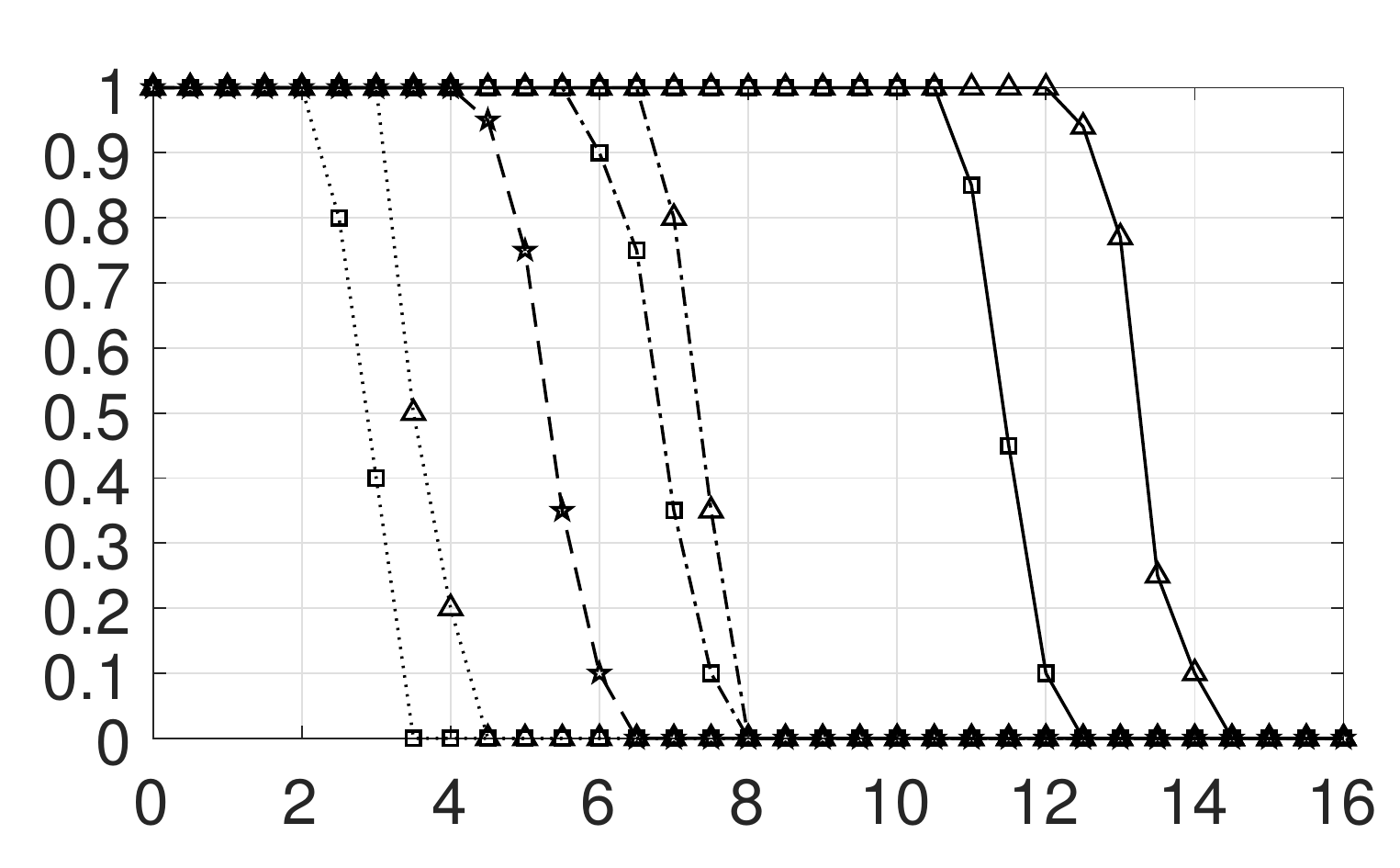}}
\subfigure[M=16, heavy per-task utilization.]{\label{fig:exp6}
\includegraphics[width=0.66\columnwidth]{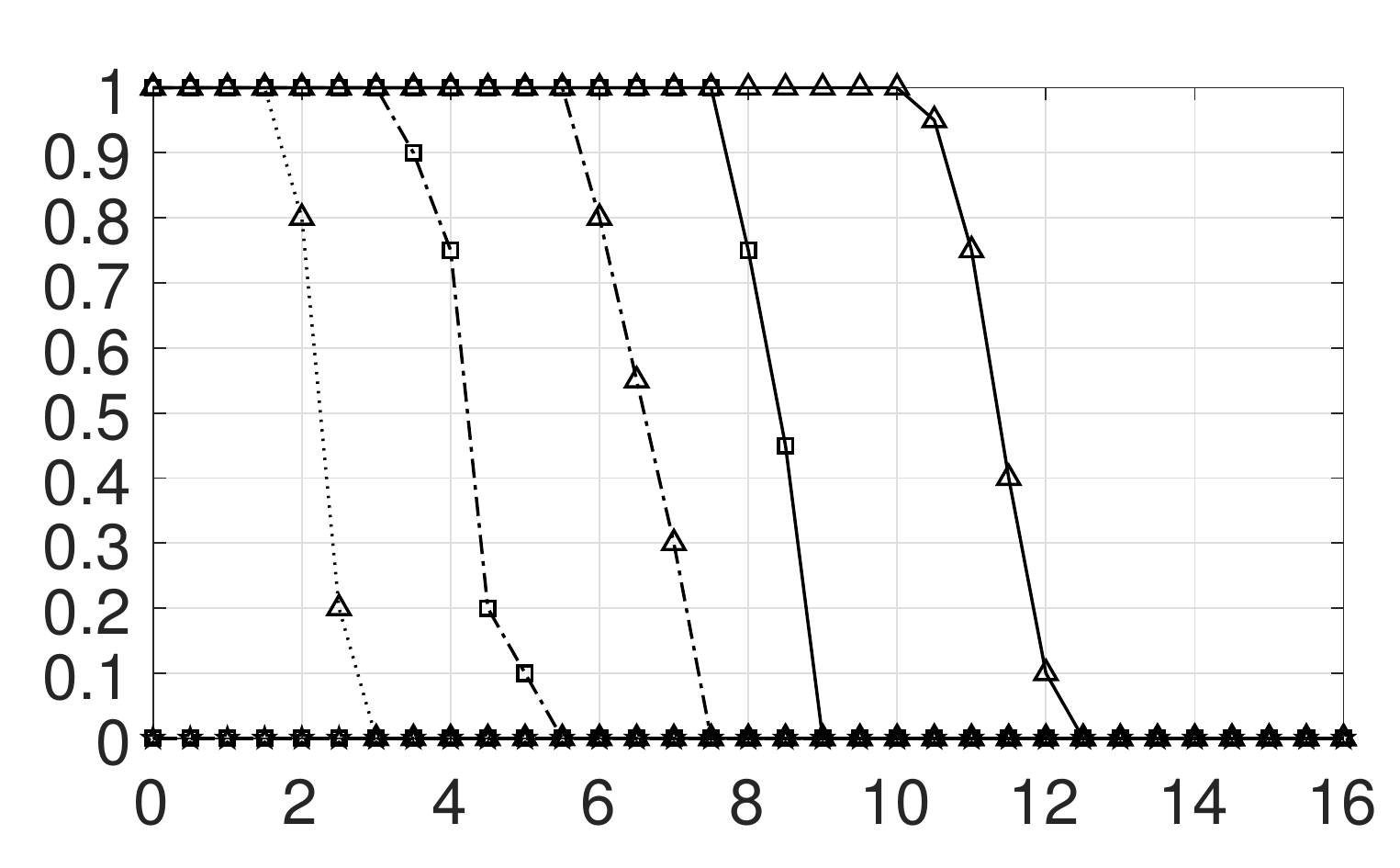}}
\caption{\footnotesize{Schedulability results. In all graphs, the x-axis represents the task set utilization cap and the y-axis represents the fraction of generated task sets that were schedulable. In the first (respectively, second) rows of graphs, M = 8 (respectively, M = 16) is assumed. In the first (respectively, second and third) column of graphs, light (respectively, medium and heavy) pertask utilizations are assumed. Each graph gives seven curves: two curves per tested approach for DAG task sets for the cases of short and long critical path, respectively, and one additional curve of the density test for ordinary sporadic task sets. As seen at the top of the figure, the label ``CP-GEDF-s(l)'' indicates the approach of our schedulability test given in Theorem~\ref{Theorem:test} assuming short (long) critical path. Similarly, ``SU-s(l)'' and ``CAB-s(l)'' labels denote the utilization-based schedulability tests given in~\cite{chen2014capacity} and~\cite{li2013outstanding} respectively.}}\label{exp:all}
\end{figure*}

\vspace{1mm}
\noindent\textbf{Relationship to the classical density test~\cite{goossens2003priority}.} If each DAG has only one single subtask, then each sporadic DAG tasks becomes an ordinary sporadic task. In this case, we have $\sigma_i = u_i$ $(1\leq i \leq n)$ and $\sigma_{max}(\tau) = u_{max}(\tau)$, where $u_{max}(\tau) = \max_{\tau_i\in\tau}u_i$. The schedulability test given in Theorem~\ref{Theorem:test} becomes identical to the classical density test.


\begin{lemma}
A set $\tau = \{\tau_1, \dots, \tau_n\}$ of $n$ independent sporadic DAG tasks is schedulable on $M$ identical processors under CP-GEDF, if each DAG task only contains a single subtask and
\begin{equation}\label{eq:dencity}
U_{sum} \leq M - (M-1)\times u_{max}(\tau)
\end{equation}
holds.
\end{lemma}

\begin{proof}
We prove this lemma by contradiction. Suppose there is a deadline miss in the schedule. We will show that this leads to a contradiction of Eq~\ref{eq:dencity}. Since each DAG task only contains a single subtask, $\sigma_i = u_i$ $(1\leq i \leq n)$ holds.


Let $\tau_h$ be the first task to miss a deadline at $t_d$, and $[t_d - \Delta, t_d)$ be the maximal $u_h$-busy window with respect to $\tau_h$, which is guaranteed by Lemma~\ref{lemma:existence} since there is a deadline miss. Since $u_h \leq u_{max}(\tau)$, we have
\begin{equation}\label{eq:umax}
M - (M-1)\times u_h \geq M - (M-1)\times u_{max}(\tau).
\end{equation}
Because $[t_d - \Delta, t_d)$ is $u_h$-busy, according to Def.~\ref{def:busy}, Def.~\ref{def:maximal} and Eq.~\ref{eq:umax}, the maximal $u_{max}(\tau)$-busy interval $[t_d - \Delta', t_d)$ on $\mathcal{S}$ exists, where $\Delta' \geq \Delta$. By Def.~\ref{def:maximal}, we have $\frac{W}{\Delta'} > M\times(1-u_{max}(\tau)) + u_{max}(\tau)$, where $W$ is the workload within $[t_d - \Delta', t_d)$. By lemma~\ref{lemma:upperbound}, $\frac{W_i}{\Delta'}\leq \eta_i$ holds w.r.t. the maximal $u_{max}(\tau)$-busy interval $[t_d - \Delta', t_d)$, for $1\leq i \leq n$. We thus have
\begin{equation}\label{eq:finaleq}
M - (M - 1)\times u_{max}(\tau) < \frac{W}{\Delta'} \leq \sum_{i=1}^n \frac{W_i}{\Delta'}\leq \sum_{i=1}^n \eta_i.
\end{equation}

Since $u_{max}(\tau) \geq u_i$ holds for every task $\tau_i$, in Eq.~\ref{eq:finaleq} $\eta_i = u_i$ by Lemma~\ref{lemma:upperbound}. We have
\begin{equation}
M - (M - 1)\times u_{max}(\tau) <  \sum_{i=1}^n \eta_i = U_{sum},
\end{equation}
which contradicts Eq.~\ref{eq:dencity}.
\end{proof}

\section{Experiments}

We have conducted extensive sets of experiments using randomly-generated DAG task sets to evaluate the applicability of Theorem~\ref{Theorem:test}. We compare our test denoted ``CP-GEDF'' with the only two existing utilization-based schedulability tests: one denoted by ``CAB'' given by Corollary~6 in \cite{li2013outstanding}, and the second one denoted by ``SU'' given by Corollary~2 in \cite{chen2014capacity}.

\vspace{1mm}
\noindent\textbf{Experimental setup:} In our experiments, each DAG task set was generated randomly as follows, which is similar to the task set generation methods used in~\cite{liu2012m, dong2017analysis}. Task periods were uniformly distributed over $[50ms,200ms]$. Task utilizations were distributed differently for each experiment using three uniform distributions. The ranges for the uniform distributions were $[0.005, 0.5]$ (light), $[0.5, 1]$ (medium), and $[1, 1.5]$ (heavy). Each DAG's execution time was calculated from the corresponding period and utilization values. The cp-utilization of each DAG was generated using two uniform distributions: $[0.1\times u_i, 0.3 \times u_i]$ (critical paths are relatively short), and $[0.3\times u_i, 0.5\times u_i]$ (critical paths are relatively long). We varied the total system utilization $U_{sum}$ within $\{0.1, 0.2, \dots, M\}$. For each combination of per-DAG task utilization, cp-utilization, and $U_{sum}$, 1,000 task sets were generated for systems with $8$ or $16$ processors. Each such task set was generated by creating tasks until total utilization exceeding the corresponding utilization cap, and by then increasing the last task's period so that the total utilization equals the utilization cap. For each generated set, HRT schedulability was checked under CP-GEDF, CAB, and SU.

Another interesting comparison examined in our experiments is to verify whether our derived test takes advantage of the intra-DAG prallelism. That is, whether does the intra-DAG parallelism help a DAG task system easier to be schedulable on a multiprocessor compared to ordinary sporadic task systems. For this evaluation, for each generated DAG set, we generate a corresponding sporadic task set where each sporadic task $\tau_i$ has the same utilization as the utilization of each DAG task $\tau_i$ (the total utilization of subtasks of the DAG).
 We check the schedulability of each such generated sporadic task set by the classic density test~\cite{goossens2003priority}, denoted as ``U-ordinary''. Note that if any generated DAG task in an experiment set has a utilization greater than one, then we choose not to compare CP-GEDF with U-ordinary for that experiment set.

\vspace{1mm}
\noindent\textbf{Experimental results:} The experimental schedulability results are shown in Fig.~\ref{exp:all} (the organization of each sub-figure is described in the figure's caption). Each curve represents the fraction of the generated task sets successfully scheduled by the corresponding approach, as a function of the total system utilization. As seen in the figure, under all tested scenarios, CP-GEDF improves upon SU and CAB by a notable margin. For example, in Fig.~\ref{fig:exp3}, when DAGs' critical paths are short and per-task utilization is medium, CP-GEDF-s can achieve 100\% schedulability when $U_{sum}$ is not exceeding $6$ while SU and CAB tests fail to do so when $U_{sum}$ merely exceeds $3.25$ and $1.25$ ,respectively.\footnote{Let CP-GEDF-s (CP-GEDF-l) denote the case with short (long) critical paths.} Note that in Fig.~\ref{fig:exp5} and Fig.~\ref{fig:exp6}, when tasks' critical paths are long and per-task utilization is heavy, no task set can pass the CAB test.

The experimental results also verify that our schedulability test explores the intuition where the DAG task model may better benefit from the multiprocessor parallelism compared to the sproadic task model. As seen in Fig.~\ref{fig:exp1}, Fig.~\ref{fig:exp3}, Fig.~\ref{fig:exp2} and Fig.~\ref{fig:exp4}, CP-GEDF improves upon U-ordinary by rather large margins. This is because with the same per-task utilization, the cp-utilization for each DAG task is smaller than the task utilization for each ordinary sporadic task. According to our test and the classical density test, it is easier for the DAG task set to pass the schedulability test (Theorem~\ref{Theorem:test}).

\section{Conclusion}

In this paper, we present a set of novel scheduling and analysis techniques for better supporting hard real-time sporadic DAG tasks on multiprocessors, through smartly defining and analyzing the execution order of subtasks in each DAG. As demonstrated by experimental results, our proposed test significantly improves upon existing utilization-based tests with respect to schedulability, and is often able to guarantee schedulability with little or no utilization loss.

\bibliographystyle{IEEEtran}
\bibliography{DAG}

\clearpage
\section*{Appendix}
\noindent\textbf{Proof of Lemma~\ref{lemma:lastjob}}.
\begin{proof}
Let $\mathcal{S}$ denote the CP-GEDF schedule and $f_{i,j}^{x_u}$ denote the completion time of job $\tau_{i,j}^{x_u}$. To prove this lemma, it suffices to prove that the job $\tau_{i,j}^{x_{k}}$, which is the last job of the critical path, is the last completed job among all jobs belonging to dag-job $\tau_{i,j}$. We prove this new proof obligation by contradiction. Assume that another job $\tau_{i,j}^{y_0}$ which completes later than $\tau_{i,j}^{x_{k}}$. 

Since under GEDF, dag-jobs released by different DAG tasks have distinct priorities and all jobs belonging to a dag-job inherit the same priority of the dag-job, the relative execution ordering of all jobs belonging to dag-job $\tau_{i,j}$ including $\tau_{i,j}^{x_{u}}$ and $\tau_{i,j}^{y_0}$  does not depend on other dag-jobs, but solely depends on the DAG structure of $\tau_i$ and the \cp policy.

We analyze the interval $[r_{i,j}, f_{i,j}^{y_0})$ by dividing it into $w \geq 1$ time intervals, denoted by $[f_{i,j}^{y_1}, f_{i,j}^{y_0}),  [f_{i,j}^{y_2}, f_{i,j}^{y_1}), \dots, [r_{i,j}, f_{i,j}^{y_w})$, ordered from right to left with respect to time. We identify these time intervals by moving from right to left with respect to time in the schedule $\mathcal{S}$ considering jobs belonging to the dag-job $\tau_{i,j}$.

\begin{figure*}[!tp]
\centerline
{\includegraphics[width=0.75\textwidth]{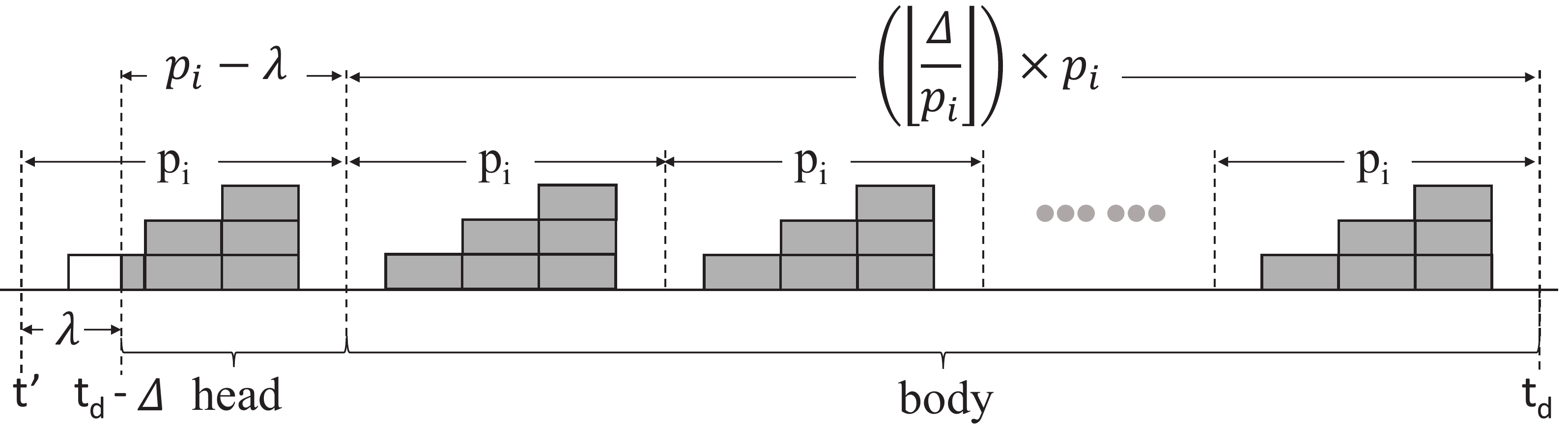}
}\caption{\footnotesize{Maximal $W_i$ for the problem window $[t_d-\Delta, t_d)$.}}\label{worstcaseWi}
\vspace{-10pt}
\end{figure*}

We identify this interval set by finding a path in $\tau_i$ starting from a source subtask and ending at the subtask $\tau_i^{y_0}$.
Moving from the time instant $f_{i,j}^{y_0}$ to the left in $\mathcal{S}$, let $f_{i,j}^{y_1}$ denote the latest completion time among $\tau_{i,j}^{y_0}$'s predecessor jobs. Note that  $\tau_{i,j}^{y_0}$ becomes ready at $f_{i,j}^{y_1}$. We thus identify the first interval$[f_{i,j}^{y_1}, f_{i,j}^{y_0})$ in this set. Moving from  $f_{i,j}^{y_1}$ to the left in $\mathcal{S}$, we apply this same process to identify the remaining intervals in this set, until we find the source subtask $\tau_{i}^{y_w}$ for this path, which is ready at the release time of the corresponding dag-job $\tau_{i,j}$ at $r_{i,j}$ and completes at $f_{i,j}^{y_w}$. Note that this path always exists because $\tau_{i,j}^{y_0}$ exists.

We now show a contradiction that there exists another path consisting of $\tau_{i}^{y_w},  \dots, \tau_{i}^{y_2}, \tau_{i}^{y_1}, \tau_{i}^{y_0}$ that is longer than the critical path of $\tau_i$. For each time interval $[f_{i,j}^{y_{q+1}}, f_{i,j}^{y_q})$ $(1\leq q\leq w-1)$, we know that $\tau_{i,j}^{y_q}$ is ready at the beginning of this interval $f_{i,j}^{y_{q+1}}$ according to the definition of the interval. Thus, each interval $[f_{i,j}^{y_{q+1}}, f_{i,j}^{y_q})$ (including $[r_{i,j}, f_{i,j}^{y_w})$) only consists of two kinds of subintervals: (\textit{i}) subintervals during which $\tau_{i,j}^{y_q}$ executes continuously, and (\textit{ii}) subintervals during which $\tau_{i,j}^{y_q}$ are not executing (i.e., being preempted by jobs belonging to other dag-jobs with higher priorities than $\tau_{i,j}$ or being delayed by jobs belonging to the same dag-job). During the first kind of subintervals, jobs from the subtasks on the critical path of $\tau_i$ may execute; but during the second kind of subintervals, jobs from the subtasks on the critical path of $\tau_i$ do not execute due to the \cp policy, for otherwise $\tau_{i,j}^{y_q}$ should have been executing during such subintervals. Thus, during each $[f_{i,j}^{y_{q+1}}, f_{i,j}^{y_q})$, the workload executed due to $\tau_{i,j}^{y_q}$ is at least the workload due to jobs from the subtasks on the critical path of $\tau_i$. Moreover, within the last interval $[f_{i,j}^{y_{1}}, f_{i,j}^{y_0})$, since $\tau_{i,j}^{y_0}$ completes later than $\tau_{i,j}^{x_k}$, we know that the workload executed due to $\tau_{i,j}^{y_0}$ must be strictly greater than the workload due to jobs from the subtasks on the critical path of $\tau_i$. Therefore, the total workload executed within $[r_{i,j}, f_{i,j}^{y_{0}})$ due to the subtask set $\{\tau_{i}^{y_w},  \dots, \tau_{i}^{y_1}, \tau_{i}^{y_0}\}$ must be strictly greater than the total workload due to the subset of subtasks on the critical path of $\tau_i$.

Clearly, we have identified another path in $\tau_i$ and this path is longer than the critical path of $\tau_i$. A contradiction is reached.
\end{proof}

\noindent\textbf{Proof of Lemma~\ref{lemma:boundWi}}. ( This lemma is proved in the same manner as lemma~10 in~\cite{baker2003}.)
\begin{proof}

The workload contributed by $\tau_i$ in $[t_d-\Delta, t_d)$, denoted by $W_i$, is the sum of the workload contributed by $\tau_i$ within three sub-windows. Our method is to identify a worst-case situation, where $W_i$ achieves the largest possible value for $\tau_i$ in $[t_d-\Delta, t_d)$. Note that under CP-GEDF, the contribution of any dag-job released by $\tau_i$ in the tail sub-window is zero w.r.t. $W_i$ as such dag-jobs have deadlines later than $t_d$.
We will first consider the case where $\Delta \geq p_i$, and then the other case where $\Delta < p_i$.

\vspace{1mm}
\noindent\textbf{Case 1:}  $\Delta \geq p_i$. 
As illustrated in Fig.~\ref{worstcaseWi}, it is evident that the maximum contribution due to $\tau_i$ in the body sub-window is achieved when the dag-jobs within the body sub-window are released periodically and the last such dag-job has the deadline at $t_d$. Thus, the number of completed dag-jobs of $\tau_i$ in the body sub-window is clearly upper-bounded by $\lfloor\frac{\Delta}{p_i}\rfloor$.

According to Lemma~\ref{lemma:upperboundon}, the upper bound of the carry-in workload of $\tau_i$ is $C_i - \Delta\times\lambda$ which is a non-increasing function of $\lambda$. Therefore, the carry-in workload is maximized when $\lambda$ is minimized. As seen in Fig.~\ref{worstcaseWi}, when the body sub-window ends exactly at $t_d$, the length of the head sub-window, $p_i - \lambda$, is $\Delta - \lfloor\frac{\Delta}{p_i}\rfloor\times p_i$. In this case, we have
\begin{equation}
p_i - \lambda = \Delta - \lfloor\frac{\Delta}{p_i}\rfloor\times p_i \Rightarrow \lambda = (\lfloor\frac{\Delta}{p_i}\rfloor + 1)\times p_i - \Delta.
\end{equation}


In the following, we explain why $\lambda$ cannot be smaller than $(\lfloor\frac{\Delta}{p_i}\rfloor + 1)\times p_i - \Delta$.
According to Lemma~\ref{lemma:upperboundon}, the carry-in workload $\varepsilon_i$ of $\tau_i$ within $[t_d-\Delta, t_d)$ is upper bounded  by $C_i - \sigma_h\times\lambda$, i.e.,
\begin{equation}\label{eq:carryinupperbound}
\varepsilon_i \leq C_i - \sigma_h\times\lambda,
\end{equation}
where $\sigma_h$ is the cp-utilization of the DAG task which misses its deadline.

If we further decrease $\lambda$, it results in at most a linear increase on $\varepsilon_i$ (i.e., $\sigma_h$) according to Eq.~\ref{eq:carryinupperbound}. Thus, If we decrease $\lambda$ to any value smaller than $(\lfloor\frac{\Delta}{p_i}\rfloor + 1)\times p_i - \Delta$, it results in at most a linear increase in the contribution to the head sub-window which is at most $C_i$. However, this causes the number of released dag-jobs of $\tau_i$ to be reduced by one within the body sub-window. Thus, any further decrease on the value of $\lambda = (\lfloor\frac{\Delta}{p_i}\rfloor + 1)\times p_i - \Delta$ will cause $W_i$ to decrease. Because $[t_d-\Delta, t_d)$ includes the head sub-window and the body sub-window, when we decrease $\lambda$, the length of the head sub-window increases and the length of body sub-window decreases. When $\lambda = (\lfloor\frac{\Delta}{p_i}\rfloor + 1)\times p_i - \Delta$, the body sub-window ends exactly at $t_d$. If the length of body sub-window decreases, the deadline of the last dag-job released by $\tau_i$ before $t_d$ will exceeds $t_d$. Since every dag-job with a deadline later than $t_d$ is removed from $\mathcal{S}$, the contribution of $\tau_i$ to the body sub-window is decreased by $C_i$.

Therefore, the value of $W_i$  gets maximized when $\lambda = (\lfloor\frac{\Delta}{p_i}\rfloor + 1)\times p_i - \Delta$, and we have
\begin{equation}
W_i \leq \lfloor\frac{\Delta}{p_i}\rfloor\times C_i + \max\{0, C_i - \sigma_h\times\lambda\}.
\end{equation}

\vspace{1mm}
\noindent\textbf{Case 2:}  $\Delta < p_i$. In this case, the contribution on $W_i$ due to the body sub-window is $0$, since it is impossible for a dag-job of $\tau_i$ to have both release time and deadline within the window. 
Thus, according to Lemma~\ref{lemma:upperboundon}, $W_i$ in this case is at most $\max\{0, C_i - \sigma_h\times\lambda\}$.

By combining these two cases, we have
\begin{equation}
    W_i=
    \begin{cases}
      \lfloor\frac{\Delta}{p_i}\rfloor\times C_i + \max\{0, C_i - \sigma_h\times\lambda\}, & \text{if}\ \Delta \geq p_i \\
      \max\{0, C_i - \sigma_h\times\lambda\}, & \text{otherwise}
    \end{cases}
\end{equation}
\noindent where $\lambda = (\lfloor\frac{\Delta}{p_i}\rfloor + 1)\times p_i - \Delta$.
\end{proof}

\noindent\textbf{Proof of Lemma~\ref{lemma:upperbound}}. ( This lemma is proved in the same manner as lemma~11 in~\cite{baker2003}.)
\begin{proof}
Let $\theta(\Delta)$ denote $\frac{W_i}{\Delta}$. In this lemma, we upper bound $\theta(\Delta)$ for $\tau_i$. According to lemma~\ref{lemma:boundWi}, we have
\begin{equation}
    \theta(\Delta)= \frac{W_i}{\Delta} =
    \begin{cases}
     \frac{ \lfloor\frac{\Delta}{p_i}\rfloor\times C_i + \max\{0, C_i - \sigma_h\times\lambda\}}{\Delta}, & \text{if}\ \Delta \geq p_i \\
     \frac{ \max\{0, C_i - \sigma_h\times\lambda\}}{\Delta}, & \text{otherwise}
    \end{cases}
\end{equation}
\noindent where $\lambda = (\lfloor\frac{\Delta}{p_i}\rfloor + 1)\times p_i - \Delta$.

Since $\max\{0, C_i - \sigma_h\times\lambda\}$ may have different values, we have two cases here:

\noindent\textbf{Case 1:} $\max\{0, C_i - \sigma_h\times\lambda\}  = 0$.

We have $C_i - \sigma_h\times\lambda\leq 0$, thus $\lambda \geq \frac{C_i}{\sigma_h}$. Since we also know
that $\lambda < p_i$, we have $\sigma_h <  \frac{C_i}{p_i}$. Since $\lfloor\frac{\Delta}{p_i}\rfloor \leq \frac{\Delta}{p_i}$, we have
\begin{equation}
\theta(\Delta) = \frac{W_i}{\Delta} = \frac{\lfloor\frac{\Delta}{p_i}\rfloor\times C_i}{\Delta}\leq \frac{\frac{\Delta}{p_i}\times C_i}{\Delta} = \frac{C_i}{p_i} = u_i.
\end{equation}

\noindent\textbf{Case 2:} $\max\{0, C_i- \sigma_h\times\lambda\}  \neq 0$.

We have $C_i - \sigma_h\times\lambda > 0$. Since $\lambda = (\lfloor\frac{\Delta}{p_i}\rfloor + 1)\times p_i - \Delta$,

\begin{equation}
\begin{split}
\theta(\Delta) &= \frac{\lfloor\frac{\Delta}{p_i}\rfloor\times C_i + C_i - \sigma_h\times\lambda}{\Delta}\\
&= \frac{\lfloor\frac{\Delta}{p_i}\rfloor\times(C_i - \sigma_h\times p_i) + C_i - \sigma_h\times(p_i - \Delta)}{\Delta}.
\end{split}
\end{equation}

\noindent\textbf{Case 2.1:} $C_i - \sigma_h\times p_i > 0$, i.e.~$\sigma_h < \frac{C_i}{p_i}$.

\begin{equation}
\begin{split}
\theta(\Delta)&= \frac{\lfloor\frac{\Delta}{p_i}\rfloor\times(C_i - \sigma_h\times p_i) + C_i - \sigma_h\times(p_i - \Delta)}{\Delta}\\
&=\frac{\frac{\Delta}{p_i}\times(C_i - \sigma_h\times p_i) + C_i - \sigma_h\times(p_i - \Delta)}{\Delta}\\
&= \frac{C_i}{p_i} + \frac{C_i - \sigma_h\times p_i}{\Delta}\\
&\leq \frac{C_i}{p_i} + \frac{C_i - \sigma_h\times p_i}{p_h} = u_i + \frac{C_i - \sigma_h\times p_i}{p_h}.
\end{split}
\end{equation}

\noindent\textbf{Case 2.2:} $C_i - \sigma_h\times p_i \leq 0$, i.e.~$\sigma_h \geq \frac{C_i}{p_i}$. It is evident that $\lfloor\frac{\Delta}{p_i}\rfloor > \frac{\Delta}{p_i} - 1$, then we have
\begin{equation}
\begin{split}
\theta(\Delta)&= \frac{\lfloor\frac{\Delta}{p_i}\rfloor\times(C_i - \sigma_h\times p_i) + C_i - \sigma_h\times(p_i - \Delta)}{\Delta}\\
&< \frac{(\frac{\Delta}{p_i}-1)\times(C_i - \sigma_h\times p_i) + C_i - \sigma_h(p_i - \Delta)}{\Delta}\\
&= \frac{C_i}{p_i} = u_i.
\end{split}
\end{equation}

Thus, this lemma holds.
\end{proof}

\end{document}